\documentclass[preprint,12pt]{elsarticle}

\usepackage{amssymb}
\usepackage{enumitem}
\usepackage{hyperref}

\journal{AxRiv}

\newtheorem{proposition}{Proposition}
\newtheorem{lemma}[proposition]{Lemma}
\newtheorem{theorem}[proposition]{Theorem}
\newtheorem{corollary}[proposition]{Corollary}

\newdefinition{remark}{Remark}

\newproof{proof}{Proof}

\begin{document}

\begin{frontmatter}



\title{Counting and Enumerating Tree-Child Networks and Their Subclasses}


\author[GC]{Gabriel Cardona}\ead{gabriel.cardona@uib.es}
\author[LZ]{Louxin Zhang}\ead{matzlx@nus.edu.sg}

\address[GC]{Department of Mathematics and Computer Science, University of the Balearic Islands,\\ Ctra. Valldemossa km 7,5. E-07120 Palma, Spain}
\address[LZ]{Department of Mathematics, National University of Singapore,\\
10 Kent Ridge Road, Singapore 119076}

\begin{abstract}
Galled trees are studied as a recombination model in population genetics.  This class of phylogenetic networks is generalized into tree-child, galled and reticulation-visible network classes by relaxing a structural condition imposed on galled trees.   We provide a solution to an open problem that is how to count and enumerate  tree-child networks.   Explicit counting formulas are also given  for galled trees  through their relationship to ordered trees and phylogenetic networks in which the child of each reticulation is a leaf. 
\end{abstract}




\end{frontmatter}


%\documentclass[12pt]{article}
%\documentclass[review,onefignum,onetabnum]{siamart190516}

\section{Introduction}
Phylogenetic networks have been used often to model gene and genome evolution over recent years \cite{Gusfield_book, Lake_99}. A rooted phylogenetic network (RPN) is a rooted directed acyclic graph (DAG) where all the sink nodes are of indegree 1, called the leaves, and where there is a unique source node,  called the root.  The leaves represent a set of taxa (e.g, species, genes or individuals in a population) and the root represents the least common ancestor of the taxa.  Moreover,  in a RPN,  non-leaf and non-root nodes are divided into two classes: tree nodes, which have more children  than parents,  and reticulations,  which have more parents than children. 

It is a great challenge to reconstruct phylogenetic networks from sequence or gene tree data \cite{Huson_09}.
 Imposing  topological conditions on networks allows us to define simple classes of RPNs, on which evolution can hopefully be understood well \cite{Gusfield_book,Genetics19,Lesnick_18}. One network class is tree-child networks (TCNs) \cite{Cardona_09b}, in which every non-leaf node has a child that is a tree node or a leaf. Other popular network classes include 
  galled trees \cite{Gusfield_04, Wang_01}, normal networks \cite{Willson_07}, galled networks \cite{Huson_07},  reticulation-visible networks \cite{Huson_book} and tree-based networks \cite{Francis_15,Zhang_16} (see also \cite{Steel_book, Zhang_18}).    Indeed, the tree and cluster containment problems that are NP-complete in general   become solvable in polynomial time when restricted to all but tree-based network class \cite{Bordewich_16, Gambette_15,Gambette_18,Gunawan_16}.

In this paper, we investigate how to count  TCNs and other related classes. Phylogenetic trees are RPNs with no reticulations. It is known that there are  $2^{1-n}(2n-2)!/(n-1)!$ binary phylogenetic trees on $n$ taxa. However, counting becomes much harder for  RPNs \cite{Fuchs_18,Semple_15}.  Recently, progresses  have been made for TCNs,  normal networks and galled trees.  Semple and Steel first studied how to count unrooted galled trees. They gave closed formulas for the number of 
unrooted galled trees in terms of two parameters: the number of galls and the total number of edges over all the gall cycles \cite{Steel_06}.   Bouvel et al. presented generating functions and explicit formulas for count of unrooted and rooted galled trees \cite{Bouvel_18}.
Chang et al. studied how to encode and compress galled trees \cite{Chang_2018}.
  McDiarmid et al. gave approximate formulas for the number
of binary  TCNs and normal networks \cite{Semple_15}. Fuchs et al. presented generating functions for the count of  labeled normal networks and TCNs with few reticulations \cite{Fuchs_18}. 
Cardona et al. developed an exhaustive method for enumerating TCNs \cite{Cardona_19}, which allows one to obtain the exact number of TCNs on six taxa. 
Additionally, the author of this paper and collaborators provided a recurrence approach for counting and enumerating galled networks on $n$ taxa \cite{Zhang_Rathin_G18}.  Unfortunately, no simple formulas are known for exactly counting TCNs 
and normal networks in terms of the number of reticulations and the number of leaves. 

We make three contributions to counting TCNs and other classes of networks. First, we establish a relationship between ordered trees and binary galled trees. Using this relationship, we  derive a different formula for the number of galled trees and the first formula for the number of normal galled trees in Section~\ref{Galled_Sec3}. Secondly, 
the concepts of tree-components and 
network compression were introduced by the author and his collaborators to study  the tree containment problem \cite{Gunawan_16, Zhang_18}. Here,  we apply these concepts to study how to count TCNs.  We present a simple closed formula for counting TCNs in which the child of each reticulation is a leaf.  We then present a recurrence formula to count TCNs through counting and enumerating their component graphs in Section~\ref{Sect4_TC}. Additionally, by using the obtained formulas, we are able to compute the exact number of TCNs on eight taxa.  
Lastly, we present explicit formulas for counting phylogenetic networks with one or two reticulations in different network classes. We conclude the study by posing several research problems about counting phylogenetic networks.

\section{Basic Notation}
\label{sec2}

\subsection{DAGs}

A {\it directed graph} consists of a finite nonempty node set ${\cal V}$ together with a specified set ${\cal E}$ of ordered pairs 
of nodes of ${\cal V}$. Each element of ${\cal E}$ is called an {\it edge}.
 A {\it DAG} is a directed graph with no loops and no directed cycles.
In this study,  two parallel edges with the same orientation may exist between two distinct nodes. 

Let $u$ and $v$ be two nodes of a DAG. If $(u, v)\in {\cal E}$,  we say that $u$ is a {\it parent} of $v$ and $v$ is a {\it child} of $u$. 
  The {\it outdegree} and {\it indegree} of $u$ are defined to be the number of children and parents of $u$, respectively.  The nodes of outdegree 0 are called the {\it leaves}. 

If there is a directed path from $u$ to $v$, $u$ is said to be an {\it ancestor} of $v$ and to be {\it above} $v$; $v$ is said to be a {\it descendant} of $u$ and to be {\it below} $u$. The nodes $u$ and $v$ are {\it incomparable} if neither is an ancestor of the other. The set of leaves below $u$ is called the {\it cluster} of $u$. Two trees are not identical if and only if  they contain different node clusters \cite{Huson_book}. 

A rooted DAG has  a node of indegree 0, called the {\it root}, which is distinguished from the others. In a {\it labeled} DAG with $n$ nodes, the integers from 1 to $n$ are assigned to the nodes, inducing a linear ordering on the nodes. In a {\it leaf-labeled} DAG with $k$ leaves, the integers from 1 to $k$ are assigned to its leaves. In a {\it ordered} DAG, the children are ordered for every non-leaf node.

Two directed graphs are {\it isomorphic} if there is a one-to-one map from the node set of
one graph onto that of the other which preserves the directed edges. The isomorphic map from a rooted (leaf-)labeled DAG to another  preserves not only the edges but also the labeling and the root. The isomorphic map from a ordered DAG to another also preserves the ordering of the children for every node.  Our object is to count non-isomorphic phylogenetic networks of different types, which are rooted DAGs used to model molecular evolution. 

\subsection{RPNs}

A binary RPN on a finite set of taxa $X$ is a rooted DAG 
with no parallel edges such that:
\begin{itemize}
   \item  the {\it root} is the unique node of indegree 0. The root is of outdegree 1. 
   \item there are exactly $|X|$ leaves that are labeled one-to-one with $X$; 
 \item  non-leaf and non-root nodes are either of indegree 2 and outdegree 1,  or 
 of indegree 1 and outdegree 2; and 
 \item each edge is directed away from the root. 
\end{itemize}
For convenience, in the rest of this study, edge orientation is omitted in the graphic representation of RPNs, as illustrated in Figure~\ref{Fig_1}. 

The nodes of indegree 2 and outdegree  1 are called {\it reticulations}.  The nodes of outdegree 2 and indegree 1 are called {\it tree nodes}. 
%We also consider the leaves as tree nodes in this study.
An edge $(u, v)\in {\cal E}(N)$ is called a {\it tree edge} if $v$ is either a tree node or a leaf; and the edge is called a {\it reticulation} edge if $v$ is a reticulation. 
%The open edge entering the root is also considered to be a tree edge. 
The following simple results will be used frequently. Here we omit its proof. 

\begin{proposition} 
\label{prop2.1}
Let $N$ be a RPN with $k$ reticulations on $n$ taxa.  $N$ has $k+n-1$ tree nodes and
$k+2n-1$ tree edges and $2k$ reticulation edges. 
\end{proposition}

We will adopt the following notation:

\begin{tabular}{rl}
$[n]$ & the set $\{1, 2, \cdots, n\}$, where $n$ is a positive integer;\\
 $N$ & a RPN on $[n]$;\\
${\cal V}(N)$ &  the set of nodes of a RPN $N$;\\
${\cal R}(N)$ &  the set of reticulations of $N$;\\
${\cal T}(N)$  & the set of tree nodes of $N$;\\
${\cal E}(N)$ & the set of  edges of $N$;\\
\end{tabular}

\subsection{Network classes}
\label{sec2.2}

A binary {\it phylogenetic tree}  is simply a binary RPN with no reticulations.

A RPN is said to be a {\it galled tree} if every reticulation $r$ has an ancestor $a_r$ such that (i) $a_r$ is a tree node, (ii) there are two edge-disjoint directed paths from $a_r$ to $r$ that form a cycle $C_r$ (if edge orientation is ignored) with $a_r$ on top and $r$ at the  bottom,  and (iii) $C_r$ and $C_s$,   as indicated in (ii),   are node-disjoint for different $r$ and $s$. The cycles associated with reticulations under this definition will be called {\it galls} in this study.
Figure~\ref{Fig_1}A shows a galled tree. Note that phylogenetic trees are galled trees. 
Every RPN with only one reticulation is also a galled tree.

A RPN is said to be a  {\it TCN} if every non-leaf node has a child that is either a tree node or a leaf (Figure~\ref{Fig_1}B). Obviously, a RPN is a TCN if and only if for each non-leaf node, there is a path from it to some leaf that consists of only tree edges.

A RPN is said to  be a {\it normal} network if it is a TCN and the two parents of $r$ are incomparable for every reticulation $r$. Figure~\ref{Fig_1}C presents a normal network, whereas the TCN in Figure~\ref{Fig_1}B is not a  normal network. 

A RPN is said to be a {\it galled network} if each reticulation $r$ has an ancestor $a_r$ such that 
Conditions (i) and (ii) under the definition of galled trees are true and (iii) all the edges of the gall $C_r$ are tree edges except for the two edges entering  $r$. The networks in Figure~\ref{Fig_1}C and D are galled networks. Clearly,  galled trees are galled networks.

A RPN is said to be a {\it reticulation-visible} network if for each reticulation $r$, there is a leaf $\ell$ such that every path from the network root to $\ell$ contains $r$.  Figure~\ref{Fig_1}D shows a reticulation-visible network that is neither a galled network nor a TCN.

A RPN is said to be a {\it tree-based} network if it can be obtained from a phylogenetic tree by recursively inserting edges between nodes that subdivide two tree edges in the obtained network in each step. 

A RPN is said to be {\it one-component}  if the child of each reticulation is a leaf.  The RPNs in Figure~\ref{Fig_1}A,  \ref{Fig_1}C and \ref{Fig_1}D are one-component networks, whereas the others are not.

Throughout this paper,  we adopt the following notation:

\begin{tabular}{rl}
%$[n]$ & the set $\{1, 2, \cdots, n\}$, where $n$ is a positive integer;\\
% $N$ & a RPN on $[n]$;\\
%${\cal V}(N)$ &  the set of nodes of a RPN $N$;\\
%${\cal R}(N)$ &  the set of reticulations of $N$;\\
%${\cal T}(N)$  & the set of tree nodes of $N$;\\
%${\cal E}(N)$ &the set of directed edges of $N$;\\
${\cal RPN}_{n, k}$ & the set of RPNs with $k$ reticulations on $[n]$;\\
${\cal GT}_{n, k}$  & the set of galled trees with $k$ reticulations on $[n]$;\\
${\cal NN}_{n, k}$  & the set of normal networks with $k$ reticulations on $[n]$;\\
${\cal RV}_{n, k}$   & the set of reticulation-visible networks with $k$ reticulations on $[n]$;\\
${\cal GN}_{n, k}$   & the set of galled networks with $k$ reticulations on $[n]$; \\
${\cal TC}_{n, k}$ & the set of tree-child networks with $k$ reticulations on $[n]$; \\
${\cal TB}_{n, k}$ & the set of tree-based networks with $k$ reticulations on $[n]$;\\
1-${\cal C}$ & the set of one-component networks in $\cal C$ for each network class $\cal C$.
\end{tabular}

\begin{figure}[t!]
            \centering
            \includegraphics[scale = 0.8]{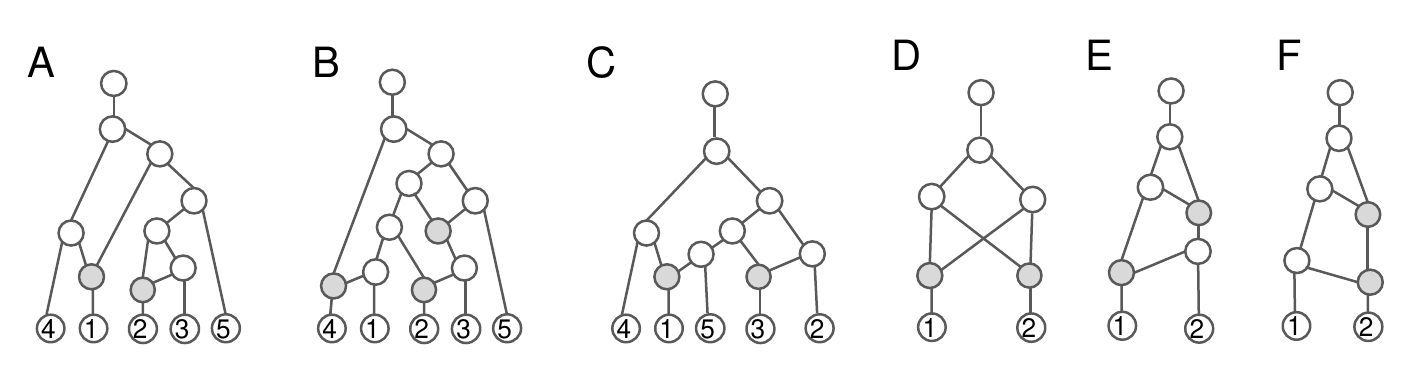}
           \caption{  {\bf Three RPNs}. {\bf A}. A galled tree.  {\bf B}. A tree-child network. {\bf C}. A normal network. {\bf D} A galled network. {\bf E} A reticulation-visible network. {\bf F} A tree-based network.  Here,  the root is the node on the top.  Tree nodes and reticulations are drawn as  unfilled and filled circles, respectively.}
            \label{Fig_1}
\end{figure}

\subsection{Tree-components and network decomposition}

Consider a RPN $N$.  
%Recall that ${\cal R}(N)$ denotes the sets of reticulations of $N$. 
Let  $N-\mathcal{R}(N)$ denote the subnetwork that is obtained from 
$N$ by the removal of all reticulations together with the incident edges. This subnetwork is actually a forest in which each connected component consists only of tree nodes and is rooted at either the network root or the child of a reticulation. 
Each of these connected components  is called a {\it tree-component} of $N$ \cite{Gunawan_18, Zhang_18}.  

Tree-component is a useful concept for characterizing the topological structures of RPNs. 
A reticulation is {\it inner} if its two parents are  in a common tree-component. It is known that each reticulation is inner for a galled network. It is also known that  every tree-component of a RPN contains  a leaf or 
the two parents of a reticulation if the network is reticulation-visible (see \cite{Zhang_18}).

It is also easy to see that a one-component network has only one non-trivial tree-component that contains all the tree nodes.

\begin{figure}[b!]
            \centering
            \includegraphics[scale = 1.0]{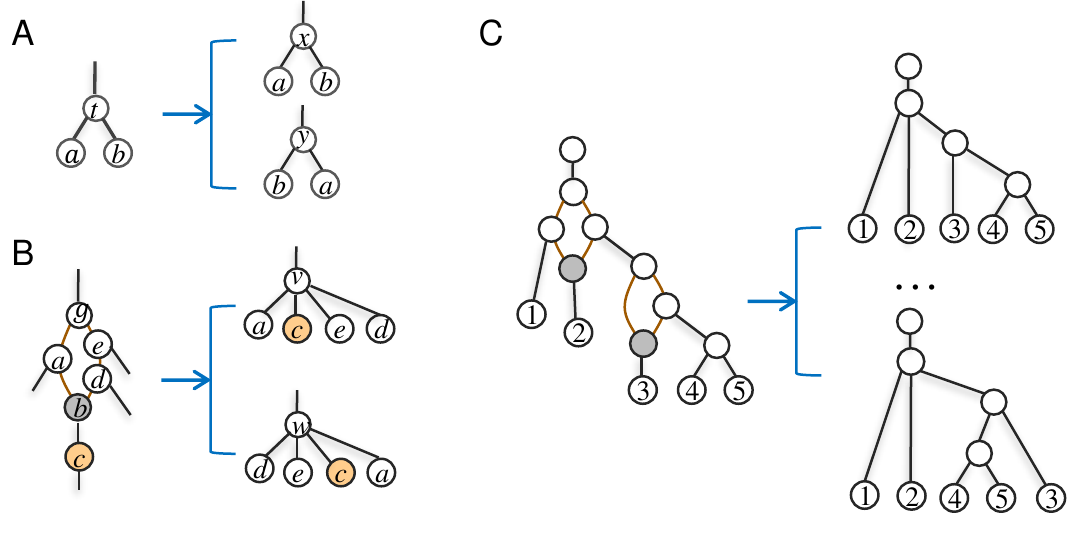}
           \caption{{\bf Illustration of transformation from a galled tree to a rooted, ordered tree with same labled leaves}. {\bf A} and {\bf B}. A mapping from a tree node to  two ordered tree nodes. {\bf C.} Mapping a galled tree to eight rooted, ordered tree with same labeled leaves by using the rules in A and B.   }
            \label{map}
\end{figure}

\section{Counting Galled Trees} 
\label{Galled_Sec3}

Let  ${\cal GT}_n$ denote the set of galled trees on $[n]$ and let ${\cal O}_n$ be the set 
of rooted, ordered trees on $[n]$.  Galled trees are also called level-1 networks. Bouvel et al. gave a closed formulas for the count of rooted and unrooted galled trees \cite{Bouvel_18}.
In this section, we shall count (normal) galled trees on $n$ taxa by establishing a many-to-many relation $m\subseteq {\cal O}_n \times {\cal GT}_n$.  

Let $N\in {\cal GT}_n$. 
For a tree node of $N$ that is not on  any  gall, we let its children be $a$ and $b$. It is mapped to
 $m(t)=\{x, y\}$, where $x$ (resp. $y$) is a tree node with ranked children $a$  and $b$ (resp. $b$ and $a$).
This is illustrated in Figure~\ref{map}A.
 
For a tree node $g$ that is on the top of a gall $C$, we let $b$ be the unique reticulation at the bottom of $C$ and $c$ be the child of $b$.  Assuming that the  tree nodes are $t'_1, t'_2, \cdots, t'_p$ on one side of $C$ (clockwise from $b$ to $g$)  and $t''_1, t''_2, \cdots, t''_q$ (clockwise from $g$ to $b$) on the other side,  $g$ is mapped to 
 $m(g)=\{v, w\}$, where $v$ is a tree node with ranked children 
 $t'_1, t'_2, \cdots, t'_p, c, t''_1, t''_2, \cdots, t''_{q}$, whereas $w$ has the ranked children \\
$t''_1, t''_2, \cdots, t''_q, c, t'_1, t'_2, \cdots, t'_{p}$, as illustrated in Figure~\ref{map}B. 

Using the above rules, we map a galled tree to a set of ordered trees with the same labeled leaves.
Figure~\ref{map}C shows how to transform a galled tree into a set of eight rooted,  ordered trees with leaves labeled with integers in $[n]$.  The following fact is clearly true.

\begin{lemma}
  Let $N\in {\cal GT}_n$ have $r$ reticulations and $k$ tree nodes that are not on any galls associated with reticulations. Then, replacing each of the $r+k$ tree nodes $t$ that  are either on the top of a gall or not in any gall
with either image of $m(t)$ produces $2^{r+k}$ trees of ${\cal O}_n$. 
\end{lemma}

Conversely, we can derive  a set of galled trees on $[n]$ from a rooted, ordered tree on $[n]$ by transforming 
a non-leaf node with $d$ ordered children into one of $d$ galls for $d>2$,  and a non-leaf node 
into a galled or an unordered tree node if $d=2$, as shown in Figure~\ref{reverse_map}.  Thus, we have the following 
lemma.

\begin{lemma}
Let $T\in {\cal O}_n$  have $s$ nodes,  each having two ordered children,  and  $t$  nodes, each having $m_1, m_2, \cdots, m_t$ ordered children, where $m_i> 2$ for each $i$.   $T$ can then be transformed into 
$3^sm_1m_2\cdots m_t$ different galled trees on $[n]$.
\end{lemma}

\begin{figure}[b!]
            \centering
            \includegraphics[scale =1]{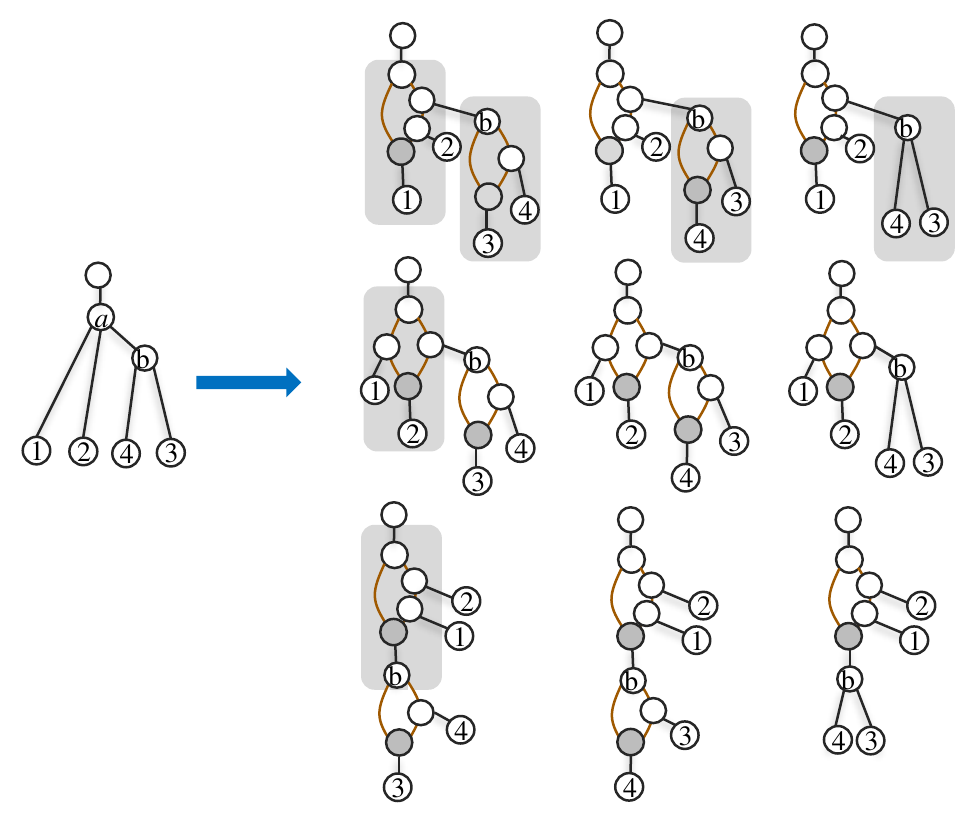}
           \caption{{\bf Illustration of transformation from a rooted, ordered tree to a set of galled trees.} The ordered tree has  the internal nodes $a$ and $b$. The node $a$ has three children ordered from left to right. It can be mapped to one of the three possible galls (shaded in the first column). The node  $b$ has two children ordered from left to right. The node can remain as a binary node or be mapped to one of the two possible galls (shaded in the first row). Therefore, the ordered tree corresponds to  nine galled trees.}
            \label{reverse_map}
\end{figure}

\begin{theorem}
Let $$C=\{ (k_2, k_3, \cdots, k_n) \;|\;  n=1+k_2 +2k_3 + \cdots + (n-1)k_{n}; k_i\geq 0, i=2, \cdots, n \}.$$ 
 We then have:
   \begin{eqnarray}
\label{galled_no2}
     %|{\cal GT}_{n, k}|
     \sum_{(k_2, k_3, \cdots, k_{n})\in C}
     \frac{(n+k_2+\cdots + k_{n}-1)!3^{k_2+k_3}4^{k_4}\cdots n^{k_n}}
      {k_2!k_3!\cdots k_n! 2^{k_2+k_3+\cdots + k_n}}.
\end{eqnarray} 
galled trees on $[n]$.
\end{theorem}
\begin{proof}
For any $(k_2, k_3, \cdots, k_n)\in C$, by a theorem of Erd\H{o}s and Sz\'{e}kely \cite[Theorem 1]{Erdos_1},
there are 
$$A=\frac{(n+k_2+\cdots +k_n -1)!}{k_2!(2!)^{k_2}k_3!(3!)^{k_3}\cdots k_n!(n!)^{k_n}}$$
rooted trees on $[n]$,  with $k_i$ internal nodes having $i$ children for $i=2, \cdots, n$.  For each internal node with $t$ children, these children can be ordered in $t!$ different ways. Thus, there are 
$$B=(2!)^{k_2}(3!)^{k_3}\cdots (n!)^{k_n}A=\frac{(n+k_2+\cdots +k_n -1)!}{k_2!k_3!\cdots k_n!}$$ rooted, ordered trees with $k_2+k_3+\cdots +k_n$ internal nodes of the prescribed degrees. By Lemma 2, we can obtained  $C=3^{k_2}3^{k_3}4^{k_4}\cdots n^{k_n}B$ galled trees on $[n]$ from the $B$ rooted, ordered trees. 

By Lemma 1, each galled trees can be obtained from
 $D=2^{k_2+k_3+\cdots+k_n}$
 rooted, ordered trees. Furthermore,  we derive Eqn.~(\ref{galled_no2}) from  dividing $C$ by $D$.\qed
\end{proof}

Since each internal node with $t\geq 3$ ordered children can be transformed into $t-2$ galls in which each side path contains at least one node, we obtain the following formula to count galled trees on $[n]$ that are also normal networks.

\begin{theorem}
Let $$C=\{ (k_2, k_3, \cdots, k_n) \;|\;  n=1+k_2 +2k_3 + \cdots + (n-1)k_{n}; k_i\geq 0, i=2, \cdots, n \}.$$ 
 There are: 
   \begin{eqnarray}
\label{galled_no}
     \sum_{(k_2, k_3, \cdots, k_{n})\in C}
     \frac{(n+k_2+\cdots + k_{n}-1)!1^{k_3}2^{k_4}\cdots (n-2)^{k_n}}
      {k_2!k_3!\cdots k_n! 2^{k_2+k_3+\cdots + k_n}}
\end{eqnarray} 
 normal galled trees on $[n]$.
\end{theorem}

\begin{figure}[t!]
            \centering
            \includegraphics[scale = 0.8]{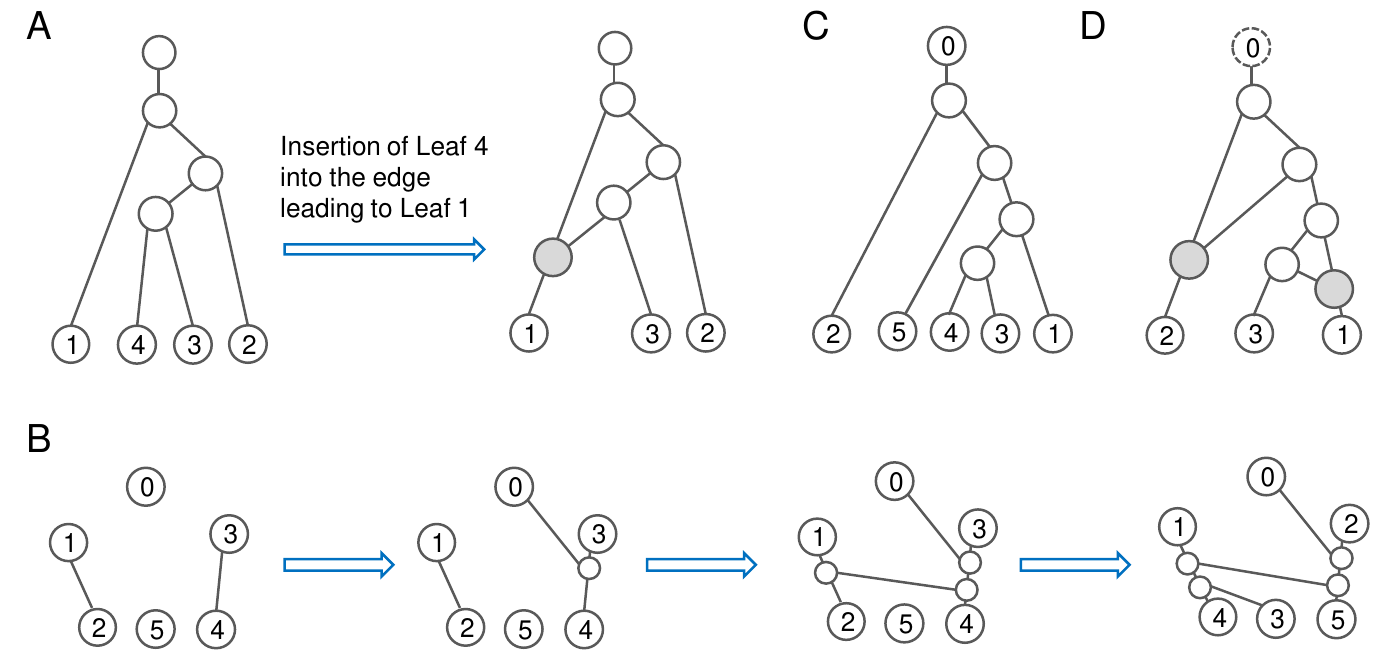}
           \caption{{\bf Illustration of the proof of Proposition~\ref{count_1cGalledTree}.} {\bf A}. Insert a leaf into the edge leading to another leaf in a tree generates a galled tree with one reticulation if the two leaves are not in a cherry. {\bf B}. Edge attachment merges two unrooted trees by inserting an edge between a node that subdivides 
an edge of one tree and a node that subdivides an edge of the other tree, where if a tree is a single node tree, the added edge is incident to the unique node of the tree. {\bf C}. The phylogenetic tree on
$[5]$ obtained from the tree built in C by rooting at 0, in which Leaves 1 and 4 are not in a cherry and Leaves 2 and 5 are not in a cherry.  {\bf D}. The galled tree that is obtained from the tree in C by  inserting Leaves 4 and 5  into the edges leading to Leaves 1 and 2, respectively. 
 \label{Ills_1C_galledTree}
}
\end{figure}

To conclude this subsection, we present a formula to count  one-component galled trees, which will be used when the galled trees with 2 reticulations are counted.
Recall that  the child of each reticulation is a leaf for one-component galled trees.

\begin{proposition} 
\label{count_1cGalledTree}
 Let ${\cal OGT}_{n, k}$ be the set of one component galled trees with $k$ reticulations on $[n]$ in which the children of the $k$ reticulation are Leaves $1$, $2$ and $k$, respectively, where $1\leq k < n$ and $n\geq 3$. Then, 
\begin{eqnarray}
  && |{\cal OGT}_{n, 1}|=\frac{(2n-2)!}{2^{n-1}(n-2)!}, \label{1ret_galledtree}\\
  && |{\cal OGT}_{n, 2}|=\frac{(2n-2)!}{3\cdot 2^{n-1}(n-3)!},\label{2ret_galledtree}\\
  &&|{\cal OGT}_{n, k}|=\frac{1}{2^{n+k-1}}\sum^{k}_{j=0}(-1)^{k-j}{k\choose j}\frac{j!(2n+2j-2)!}{(2j)!(n+j-1)!}.  
   \label{general_case}
\end{eqnarray}
where $0!=1$.
\end{proposition}
\begin{proof}
 Eqn.~(\ref{1ret_galledtree}) and (\ref{2ret_galledtree}) are special cases of Eqn.~(\ref{general_case}) where $k$ is 1 and 2, respectively. Nevertheless,  we first prove Eqn.~(\ref{1ret_galledtree}) as a warm-up exercise.  

In a RPN, two leaves are said to be in a cherry if they are adjacent to the same tree node.
Consider all $\frac{(2n)!}{2^nn!}$ possible phylogenetic trees on $[n+1]$. From each of these trees such  that 
Leaves $1$ and $n+1$ are not in a cherry, we can generate a galled tree of ${\cal OGT}_{n, 1}$ by 
inserting Leaf $n+1$ into the edge leading to Leaf 1 so that Leaf $n+1$ becomes a reticulation parent of Leaf 1 in the resulting galled tree (Figure~\ref{Ills_1C_galledTree}A). 
Since there are $\frac{(2n-2)!}{2^{n-1}(n-1)!}$ trees in which Leaves 1 and $n+1$ are in  a cherry, there are 
 $\frac{(2n)!}{2^nn!}-\frac{(2n-2)!}{2^{n-1}(n-1)!}$ trees in which the two leaves do not appear in a cherry.

For two trees $T_1$ and $T_2$, our galled tree generation procedure produces the same galled tree from them  if $T_1$ is identical to $T_2$ after Leaf 1 and Leaf $n+1$ are interchanged.   Therefore, 
 $|{\cal OGT}_{n, 1}|=\frac{1}{2}\left(\frac{(2n)!}{2^nn!}-\frac{(2n-2)!}{2^{n-1}(n-1)!}\right)=\frac{(2n-2)!}{2^{n-1}(n-2)!},$
obtaining Eqn.~(\ref{1ret_galledtree}).

In order to prove Eqn.~(\ref{general_case}), we first introduce the edge attachment operation, illustrated in Figure~\ref{Ills_1C_galledTree}B. 
 Let $T$ and $T'$ be unrooted two phylogenetic trees on $L$ and $L'$, respectively,  such that 
$L\cap L'=\emptyset$. An {\it edge attachment} operation builds a phylogenetic tree on $L\cup L'$ by 
adding an edge between a node that is either a node subdividing an edge of $T$ or the unique node of $T$ if $|{\cal V}(T)|=1$ 
and another node that is either a node subdividing an edge of $T'$ or is the node of $T'$ if $|{\cal V}(T')|=1$. Clearly, 
we can obtain a set of phylogenetic trees from a forest consisting of $k$ unrooted phylogenetic trees over distinct taxa by applying $k$ edge attachments between the original and resulting trees in such a way that each edge attachment reduces the number of  
the trees in the forest by 1. 

We use $K(i, i')$ to denote the one-edge unrooted tree with labeled leaves $i$ and $i'$ and 
$P(i)$ to denote the tree with a single node labeled with $i$. Let $S(m, r)$ be the set of unrooted phylogenetic trees  that are   
obtained by applying $m+1$ edge attachments onto a forest $F(m, r)$ consisting of $m+1$ unrooted trees $\{K(i, m+i), P(0), P(j)\;|\; 1\leq i\leq r, r+1\leq j\leq m\}$.
By Theorem 2.8.3 in \cite[page 39]{semple_book}, we obtain: 
\begin{eqnarray*}
 |S(m, r)|&=&\frac{b(m+r+1)}{b((m+r+1)-(m+1)+2)}\prod^{r}_{i=1}\left|{\cal E}(K(i, m+i))\right|\\
   &=&\frac{b(m+r+1)}{b((m+r+1)-(m+1)+2)}\\
  &=& \frac{(2m+2r-2)!r!}{2^{m-1}(m+r-1)!(2r)!},
\end{eqnarray*}
as $|{\cal E}(K(i, m+i))|=1$ and $b(k+1)=\frac{(2k-2)!}{2^{k-1}(k-1)!}$, which is the number of unrooted phylogenetic trees on $k+1$ taxa.

Let $k\leq n$. Let $T\in S(n, k)$ such that Leaves $i$ and $i+n$ are not in a cherry for each 
$i$ from 1 to $k$. We can then obtain a galled tree $G$ with $k$ reticulations on $[n]$ by rooting the tree at  Leaf 0 and  inserting Leaf 
$i+n$ into the edge leading to Leaf $i$ so that  $i+n$ becomes the reticulation parent of Leaf $i$  for each $i$, as illustrated in Figure~\ref{Ills_1C_galledTree}D.

Let $T\in S(n, k)$ in which $i$ and $i+n$ are in a cherry, where $i\leq k$.  
Let $c_i$ be the node that are adjacent to Leaves $i$ and $i+n$ in $T$.  Removing $i$ and $i+n$ together with the edges incident to them produces a unrooted tree that can also be generated by applying edge attachment 
from $F'=F(n, k) +\{P(c_i)\} -\{K(i, i+n\}$.  Conversely, we can obtain an unrooted phylogenetic tree of $S(n, k)$ in which $i$ and $i+n$ form a cherry  by attaching $i$ and $i+n$ below $c_i$ as its children from a unrooted tree that is built from $F'$ through edge attachment. For any subset $I$ of $[1, k]$,  $S(n, k)$ contains  exactly $|S(n, k-|I|)|$ unrooted phylogenetic trees in which $i$ and $i+n$ form a cherry for each $i\in I$. Since $[1, k]$ has ${k\choose j}$ subsets each containing $j$ integers, 
by Inclusion--Exclusion principle, $S(n, k)$ contains $\sum^{k}_{j=0}(-1)^{k-j}{k \choose j} \frac{(2n+2j-2)!j!}{2^{n-1}(n+j-1)!(2j)!}$ phylogenetic trees in which no pairs of $i$ 
and $i+n$ form a cherry. Since each galled tree with $k$ reticulations can be generated from $2^k$ trees of $S(n, k)$, we obtain Eqn.~(\ref{general_case}). 
\qed\end{proof}

\begin{corollary}  
\label{corollary3.6}
%Since each subset of $k$ leaves can be the children of the reticulations, 
The number of one-component galled trees with $k$ reticulations on $[n]$ is equal to ${n\choose k}|{\cal OGT}_{n, k}|$.
\end{corollary}
\begin{proof} The result is derived from the fact that each subset of $k$ taxa can be the children of the $k$ reticulations in a one-component galled tree. 
\qed\end{proof}

\section{Counting TCNs}
\label{Sect4_TC}

Recall that  ${\cal TC}_{n, k}$  denotes the set of  TCNs with $k$ reticulations on $[n]$.
We will enumerate and count the networks of ${\cal TC}_{n, k}$  through enumerating their component graphs. 

\subsection{One-component TCNs}
\label{sect5}

One-component TCNs can be considered as the building blocks of arbitrary TCNs, as we will see later. We use $\mathcal{O}_{n, k}$ to denote the set of one-component TCNs with $k$ reticulations over $[n]$.
Recall that the child of each reticulation is a leaf in a network of  $\mathcal{O}_{n, k}$.

Let $N\in \mathcal{O}_{n, k}$.  For any $r\in \mathcal{R}(N)$, $N\ominus r$ is used to denote 
the network obtained from $N$ after the removal of $r$, its child $c(r)$,  the three edges incident to $r$,  as shown in Figure~\ref{Fig_InDel}.   In general, for any $R=\{r_1, r_2, \cdots, r_k\}\subseteq {\cal R}(N)$, we define:
$$N\ominus R=(\cdots ((N\ominus r_1)\ominus r_2)\ominus \cdots \ominus r_k).$$

\begin{figure}[b!]
            \centering
            \includegraphics[scale = 0.9]{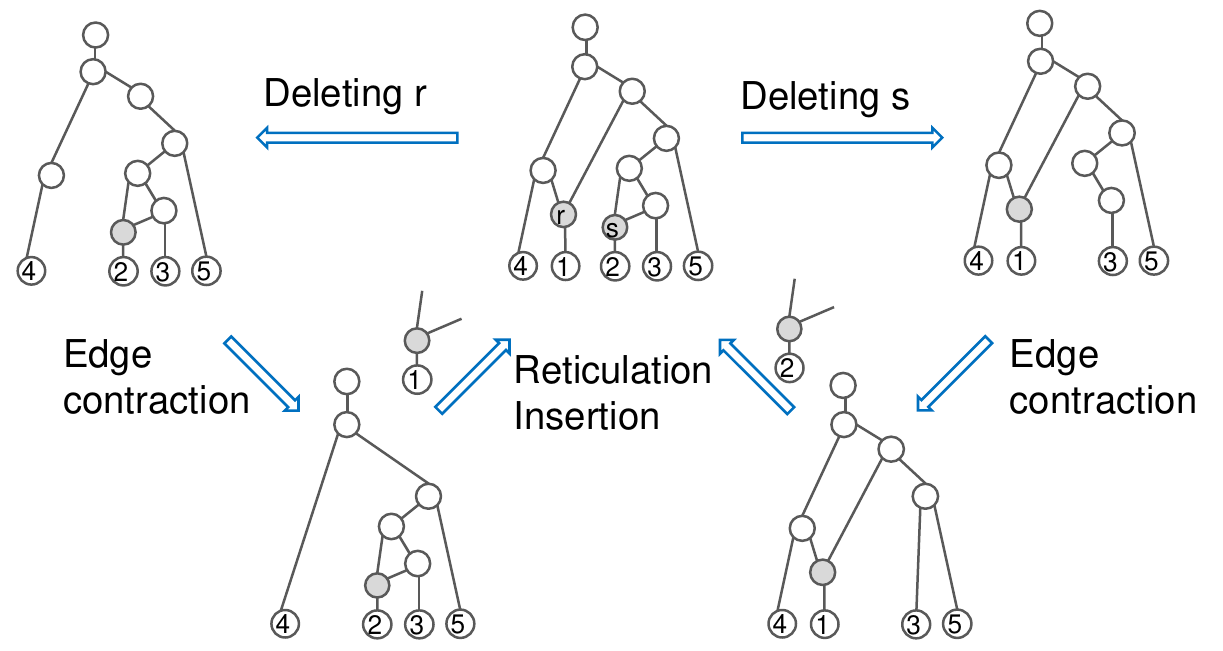}
           \caption{{\bf Illustration of reticulation insertions and deletions.} There are two types of insertions and deletions in tree-child networks: The added or removed reticulation straddles two tree edges (left) or is attached onto a single tree edge (right). }
            \label{Fig_InDel}
\end{figure}

A DAG $D$ is said to be a {\it subdivision} of a phylogenetic tree $T$, if (i)
${\cal V}(T)\subseteq {\cal V}(D)$,  (ii) each node of ${\cal V}(D)\backslash{\cal V}(T)$ is of indegree 1 and outdegree 1 and (iii) for
each $(u, v)\in {\cal E}(T)$, either $(u, v)\in {\cal E}(D)$ or there is a unique path from $u$ and $v$ passing only nodes in ${\cal V}(D)\backslash {\cal V}(T)$.

\begin{lemma}
Let $N\in \mathcal{O}_{n, k}$ such that ${\cal R}(N)=\{r_1, r_2, \cdots, r_k\}$.

\begin{enumerate}
\item For any $R\subset \mathcal{R}(N)$, $N\ominus R$ is
  also a TCN, where there may be some degree-2 nodes.

\item $N\ominus  \mathcal{R}(N)$ is the subtree $\bar{N}$ of $N$ spanned by leaves that are  the children of tree nodes. 

\item Let $T$ be the phylogenetic tree such that $N\ominus  \mathcal{R}(N)$ is a subdivision of $T$. $N$ can then be obtained from $T$ by inserting  the $k$ reticulations of ${\cal R}(N)$ together with their leaf children one by one. 
\end{enumerate}
\end{lemma}

\begin{proof}
  \begin{enumerate}
  \item \label{it1} Let $r\in {\cal R}(N)$.  The parents of $r$ are both tree
    nodes and remain in $N\ominus r$, as we only remove $r$ and its
    unique children $c(r)$ to get $N\ominus r$. Therefore, each tree
    path in $N$ remains in $N\ominus r$, implying that $N\ominus r$ is
    a TCN. By definition, we remove reticulations of $R$ one by one to
    get $N\ominus R$ and thus $N\ominus R$ is a TCN for any subset
    $R\subseteq {\cal R}(N)$.

  \item \label{it2} Let $N'=N\ominus  \mathcal{R}(N)$.  Let  $(u, v)\in {\cal E}(N')$. By \ref{it1}, $N'$ is a TCN and thus there is a path from $v$ to some leaf $\ell$ consisting of only tree edges. 
Since $N'$ does not contain any reticulation, $v$ is a tree node, $\ell$ is not the child of any reticulation in $N$. This implies that $(u, v)$ is an edge of the subtree spanned by the leaves that are the children of tree nodes in $N$. 

Conversely, we observe that the operation $\ominus$ is conducted by removing only  edges that are incident to the removed reticulation. Since any edge $e$ in the subtree $\bar{N}$ is not incident to any reticulation, $e$ is also an edge in 
$N'$. Thus, $N'=\bar{N}$.

\item The statement follows from the fact that the tail of any removed edge is a degree-2 node in $N'$ that is introduced in \ref{it2}.\qed
\end{enumerate}
\end{proof}
%$\QED$ 

Conversely, we consider  insertion of a reticulation with a leaf child into a one-component TCN. For such a TCN  $N$  on $A$ such that $A=\{a_1, a_2, \cdots, a_{n-m}\}\subset [n]$ ($m<n$),  any pair of tree edges $\{e_1, e_2\}\subset {\cal E}(N)$ and $a\not\in A$, 
we use $N(\{e_1, e_2\}, \oplus, a)$ to denote the network obtained from $N$ by  inserting a reticulation $r$,  together with its child Leaf $a$,  onto $e_1$ and $e_2$, as shown in Figure~\ref{Fig_InDel} (see \cite{Zhang_19}). Here, we allow the possibility that $e_1=e_2$.  
Formally,  let $e_1=(u_1, v_1)$ and $e_2=(u_2, v_2)$.  In this case, 
\begin{eqnarray*}
  &&{\cal V}\left(N(\{e_1, e_2\}, \oplus, a))\right)={\cal V}(N)\cup\{x_1, x_2, r, a\},\\
  && {\cal E}\left(N(\{e_1, e_2\}, \oplus, a))\right)\\
  &=&\left[{\cal E}(N)/\{e_1, e_2\}\right] \\
     & &\cup \left\{ \begin{array}{ll}
    \{(u_i, x_i), (x_i, v_i), (x_i, r), (r, a) \;|\; i=1, 2\}, & \mbox{ if } e_1\neq e_2,\\
    \{(u_1, x_1), (x_1, x_2), (x_2, v_1), (x_1, r),  (x_2, r), (r, a)\} , & \mbox{ if } e_1=e_2,
\end{array} \right.
\end{eqnarray*}
where $x_1$ and $x_2$ are the nodes that subdivide $e_1$ and $e_2$, respectively. 
\begin{lemma}  
\label{Lemma4}
Let $N$ be a one-component TCN on $A \subset [n]$.
  For any $b\in [n]\backslash A$ and four tree edges $e'_1, e'_2, e''_1, e''_2$ in $N$, 
   $N(\{e'_1, e'_2\}, \oplus, b)=N(\{e''_1, e''_2\}, \oplus, b)$ if and only if $\{e'_1, e'_2\}=\{e''_1, e''_2\}$.
 \end{lemma}
 
\begin{proof}
 Let $p'_1$ and $p'_2$ be the parents of the reticulation $r'$ that are inserted  into $e'_1$ and $e'_2$, respectively,  in $N(\{e'_1, e'_2\}, \oplus, b)$. We also let $p''_1$ and $p''_2$ be the parents of the reticulation $r''$ that are inserted into $e''_1$ and $e''_2$, respectively,  in $N(\{e''_1, e''_2\}, \oplus, b)$. 
Assume $N(\{e'_1, e'_2\}, \oplus, b)=N(\{e''_1, e''_2\}, \oplus, b)$. There is then an isomorphic map $\phi$ from 
$N(\{e'_1, e'_2\}, \oplus, b)$ to 
 $N(\{e''_1, e''_2\}, \oplus, b)$
 that preserves the edges and leaves.  Since $\phi$ maps Leaf $b$ in the former to Leaf $b$ in the latter, $\phi(r')=r''$ and thus 
$\{\phi(p'), \phi(p'_2)\}=\{p''_1, p''_2\}$. This implies that  $e'_1=e'_2$ if and only if $e''_1=e''_2$.  

Note that $p'_i$ and $p''_i$ are tree nodes and have only  a  parent  for $i=1, 2$. No matter whether  $e'_1$ and $e'_2$ are identical or not, we can further deduce that the parent and child of $p'_1$ and $p'_2$ are
mapped to the parent and child of $\phi(p'_1)$ and $\phi(p'_2)$. Therefore, $\phi$ induces an auto-isomorphic map for 
$N$. This proves that  $\{e'_1, e'_2\}=\{e''_1, e''_2\}$.
\qed\end{proof}
%\QED$ 

For any $B\subset [n]\backslash A$, we use $N\oplus B$ to denote the set of all possible TCNs obtained by inserting $k$ reticulations that each have a leaf child labeled with a unique element in $B$. Lemma~\ref{Lemma4} implies that any TCN with $k$ reticulations $r_i$ ($1\leq i\leq k$) can be obtained from a phylogenetic tree spanned by the leaves that are not below any reticulations by sequentially inserting the corresponding reticulations one by one in a unique way. Moreover, 

\begin{lemma}
Let $T'$ and $T''$ be two phylogenetic trees on $A$ such that $A\subset [n]$.
For any $B\subseteq [n]\backslash A$,  $(T' \oplus A) \cap (T'' \oplus B)\neq \emptyset$ only if $T'=T''$.
\end{lemma}
%{\bf Proof}
\begin{proof}
  Assume that $T'$ and $T''$ are distinct phylogenetic trees on $A$. 
The fact that $T'\neq T''$ implies that there is a node $u$ in $T'$ 
such that the cluster $C_u$ of $u$ is not found in $T''$ (see \cite{Steel_book} for example). 

For any $B$,  let $N\in T'\oplus B$. In $N$, 
the hard cluster $C'_u$ of $u$ is equal to $C_u\cup B'$ for some $B'\subset B$. 
If $C'_u$ appears in a TCN $M\in T''\oplus B$, then, $C_u$ appears in $M \ominus R$, where $R$ is the set of 
reticulations whose leaf children have labels in $B$.  This implies that $C_u$ also appears in $T''$, which is a contradiction. 
Therefore, $(T' \oplus B)\cap (T''\oplus B)=\emptyset$ for any $B$ if $T'$ and $T''$ are distinct. 
\qed\end{proof}
%$\QED$

\begin{theorem}
\label{thm1}
 For any $n$ and $k\leq n-1$, 
 \begin{eqnarray}
   |\mathcal{O}_{n, k}|= {n\choose k}\frac{(2n-2)!}{2^{n-1}(n-k-1)!}.
\end{eqnarray}
\end{theorem}

%\noindent {\bf Proof.}~  
\begin{proof}
  Let $N\in\mathcal{O}_{n, k}$. $N\ominus {\cal R}(N)$ is a subdivision of one of the $t_{n-k}=\frac{(2(n-k)-2)!}{(n-k-1)!2^{n-k-1}}$ rooted phylogenetic trees, each having $n-k$ leaves in the non-trivial tree-component of $N$, which contains the network root.  Let $T$ be this tree. We have that $T$ contains $2(n-k)-1$ edges (Proposition~\ref{prop2.1}) and $N$ is obtained from $T$ by inserting $k$ reticulations.

To add a reticulation  $r$  into $T$, we either insert the two parents of $r$ into an edge of $T$ or insert them onto two distinct edges. Thus, a reticulation can be added into $T$ in  $(2n-2k-1)+{2n-2k-1\choose 2}= \frac{(2n-2k-1)(2n-2k)}{2}$ ways. 

After the first reticulation is added, $T$  is  subdivided into a tree $T_1$ with $2(n-k)+1$ edges in the resulting  network that has one reticulation. Therefore, the second reticulation can be added into $T_1$ in $\frac{(2n-2k+1)(2n-2k+2)}{2}$ ways.  By induction, for any $i$,
$T$ will further be subdivided into a tree $T_{i}$ with $2(n-k)-1+2i$ edges after the first $i$ reticulations are inserted and the $(i+1)$-th reticulation can be 
inserted in $\frac{ (2n-2k+2i-1)(2n-2k+2i)}{2}$ ways. Therefore,  according to Proposition~\ref{Lemma4}, by inserting $k$ reticulations one by one in $T$, we obtain:
  $$s_{n, k}=\frac{(2(n-k)-1)(2(n-k))(2(n-k)+1)(2(n-k)+2)\dots (2n-3)(2n-2)}{2^{k}}$$
TCNs of $\mathcal{O}_{n, k}$.   
Since any $k$ out of $n$ leaves can be selected to be the children of the $k$ reticulations, 
$$|\mathcal{O}_{n, k}|={n \choose k} \times t_{n-k}\times s_{n, k}={n\choose k}\frac{(2n-2)!}{2^{n-1}(n-k-1)!}$$
and we get the formula in the statement.
\qed\end{proof}
%$\QED$

The counts of $\mathcal{O}_{n, k}$  for $1\leq k <n\leq 8$ are given in Table~\ref{Table1_1CTC}.  We remark that  $|\mathcal{O}_{n, k}|$ increases as $k$  increases for $k\leq n-\sqrt{n+1}$ and decreases  as $k$ increases for 
$k>n-\sqrt{n+1}$. This fact can be proved by considering the derivative of $|\mathcal{O}_{n, k}|/|\mathcal{O}_{n, k+1}|$.

{\small
\begin{table}
\centering 
\caption{Counts of one-component TCNs with $k$ reticulations on $[n]$,  $1\leq k <n\leq 7$. 
\label{Table1_1CTC}
}
\begin{tabular}{c|rrrrrrr}
 \hline
  $k$\textbackslash $n$ &   2 & 3 & 4 & 5 & 6 & 7 & 8 \\ 
\hline
1&  2& 18 &180   &2,100   & 28,350 &436,590 & 7,567,560 \\
2&   & 18 & 540  &12,600  &283,500 &  6,548,850 & 158,918,760 \\
3&  &      &360    &25,200 & 1,340,000  & 43,659,000 & 1,589,187,600  \\
4&  &      &        & 12,600 &  1,701,000   &130,977,000 &7,945,938,000 \\
5&  &      &       &            &    680,400   & 157,172,400 & 19,070,251,200  \\
6&  &      &       &            &                &  52,390,800  & 19,070,251,200\\
7&  &      &      &            &                &                  &  5,448,643,200\\
\hline
\end{tabular}
\end{table}
}

\subsection{Component graph}

We shall work on  component graphs to count arbitrary TCNs on $[n]$.  Let $N\in {\cal TC}_{n, k}$. Since $N$ has $k$ reticulations, it contains  $k+1$ tree-components, say, $ C_0, C_1, C_2, \cdots, C_k,$
where $C_0$ is the tree-component containing the network root and the other tree-components are each rooted at the children of the $k$ reticulations.  

Assume that the $k$ reticulations are $
\{r_i\;|\; 1\leq i\leq k\}$ and that the child of $r_i$ is the root of  the tree-component $C_i$. The {\it component graph} $G(N)$ is a direct graph that has the node  set 
$\{C_0,  C_1, \cdots, C_k\}$ and the edge set  
$\{ (C_i, C_j) \;|\;  \mbox{ $r_j$ has a parent  in $C_i$}\}.$
Since $N$ is a TCN,  the parents of $r_i$ are tree nodes for each $i$ and thus $G(N)$ is well defined and acyclic, in which the edges are oriented away from $C_0$. The component graph of the TCN in Figure~\ref{Fig_1}b is given in Figure~\ref{compG}.
 Here, we allow double edges between a pair of components $C_i$ and $C_j$, which indicates that 
 $r_j$ is an inner reticulation and its parents are both in $C_i$. 
Since $N$ is a TCN, each tree-component $C_i$  contains a subset $L_i$ of labeled leaves such that 
$L_1, L_2, \cdots, L_k$ form a partition of $[n]$. Hence, the component graphs of TCNs have a one-to-one correspondence with labeled DAGs with the property that all non-root nodes are each of indegree 2, 
 where the nodes are uniquely  labeled with the nonempty parts of a partition of $[n]$.

\begin{figure}[b!]
            \centering
            \includegraphics[scale = 1.5]{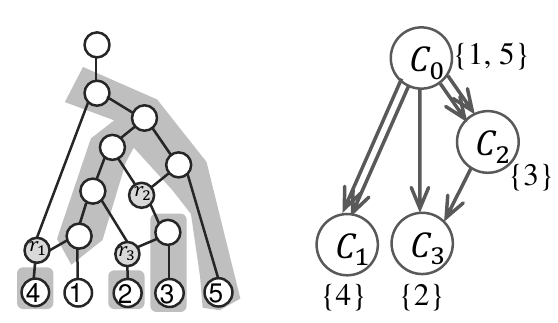}
           \caption{{\bf The component graph of the TCN in Figure~\ref{Fig_1}b}. The TCN has four components (left). $C_0$ consists of six internal tree nodes, Leaf $1$ and Leaf $5$; $C_2$ contains an internal node and Leaf $3$; $C_1$ and $C_3$  contain only a single leaf. In the component graph (right) each tree-component is labeled by the set of  leaves that appear in it.}
            \label{compG}
\end{figure}

In the rest of this subsection,  we will enumerate and count  the component graphs of TCNs as a class of rooted DAGs in which the nodes are uniquely labeled,  all nodes except  the root are of indegree 2 and two parallel edges with the same orientation between two nodes are allowed.  
%We will first introduce the notation.  

Let $G$ be a rooted DAG. The height $h(u)$ of a node $u$ is recursively defined as:
  \begin{eqnarray}
   h(u)=\left\{\begin{array}{ll}
                   0  & \mbox{ if $u$ is a leaf},\\
                   1+\max_{v: (u, v)\in {\cal E}(G)} h(v)  & \mbox{ if $u$ is a non-leaf node}.
    \end{array} \right.
\end{eqnarray}
Since $G$ is acyclic, the height of each node can be computed via  a bottom-up approach. The 
{\it level} of $G$ is defined to be one plus the height of its root, denoted by $l(G)$. For any $0\leq k< l(G)$, the $k$-th {\it row} of $G$ is defined to be  the set of the nodes of height $k$, denoted by $R_k(G)$. Note that $R_0(G)$ consists exactly of all leaves of $G$. The properties of the node height and rows are summarized as follows (the proofs have been omitted). 

\begin{proposition}
\label{height_prop}
  Let $G$ be a rooted DAG of level $l$. 

  \begin{enumerate}
  \item ${\cal V}(G) = \cup_{0\leq k<l} R_k(G)$ and
    $R_i(G)\cap R_j(G)=\emptyset$ if $i\neq j$.
  \item    Let $u\in R_k(G)$. For any edge $(u, v)\in {\cal E}(G)$, $h(v)<k$. Moreover, there is a $v$ such that $(u, v)\in {\cal E}(G)$ and $h(v)=k-1$.
  \item  There is no edge between two nodes in each row. 
\end{enumerate}
%{\rm (iv)}  For any auto-isomorphic mapping $\phi$ on $G$, $h(\phi(u))=h(u)$ for $u$ and thus
% $\{ \phi (u) \;|\; u\in R_i(G)\}=R_i(G)$ for each $i < l(G)$.
\end{proposition}

Proposition~\ref{height_prop} implies that we can construct a large rooted DAG by adding nodes row by row. For $m\geq 1$, we define ${\cal D}_{m}$ to be the set of all the rooted labeled DAGs with $m$ nodes that may have double parallel edges and satisfy the indegree constraint:
\begin{quote}
  ({\bf Indegree Constraint})   Every non-root node is of indegree 2. 
\end{quote}
 Clearly, ${\cal D}_{1}$ contains only the graph that has  an isolated node.

Let  $t$, $s$ and $m$ be positive integers such that $s<m$ and $t< m-s$.
 Consider  $G\in {\cal D}_{m-s}$ such that  $|R_0(G)|=t$.  For convenience, we set:
   $$R_0(G)=\{u_1, u_2, \cdots, u_t\}, \;\; {\cal V}(G)/R_0(G)=\{u_i\;|\; t+1\leq i\leq m-s\}.$$
 We can then extend $G$ to get different rooted DAGs $G'$  of ${\cal D}_m$ such that   $|R_0(G')|=s$ by:
  \begin{itemize}
  \item Adding  $s$ new nodes $v_1, v_2, \cdots, v_s$, and 
\item Adding two directed edges  $(u_a, v_i)$ and $(u_b, v_i)$ for each $i\in [1, s]$ such that 
there is at least an added edge leaving $u_i$ for each  $u_i\in R_0(G)$. Here,   if $u_a=u_b$, the two added edges become parallel edges between $u_a$ and $v_i$. 
\end{itemize}
 Figure~\ref{Fig4_GExtend} displays all eight possible  extensions from a graph (blue) that consists of two parallel edges from $a$ to $b$. Furthermore, Figures~\ref{Fig1_Appendix} and \ref{Fig2_Appendix} 
list  all the unlabeled component graphs with at most five nodes. 

\begin{figure}%[b!]
            \centering
            \includegraphics[scale = 0.8]{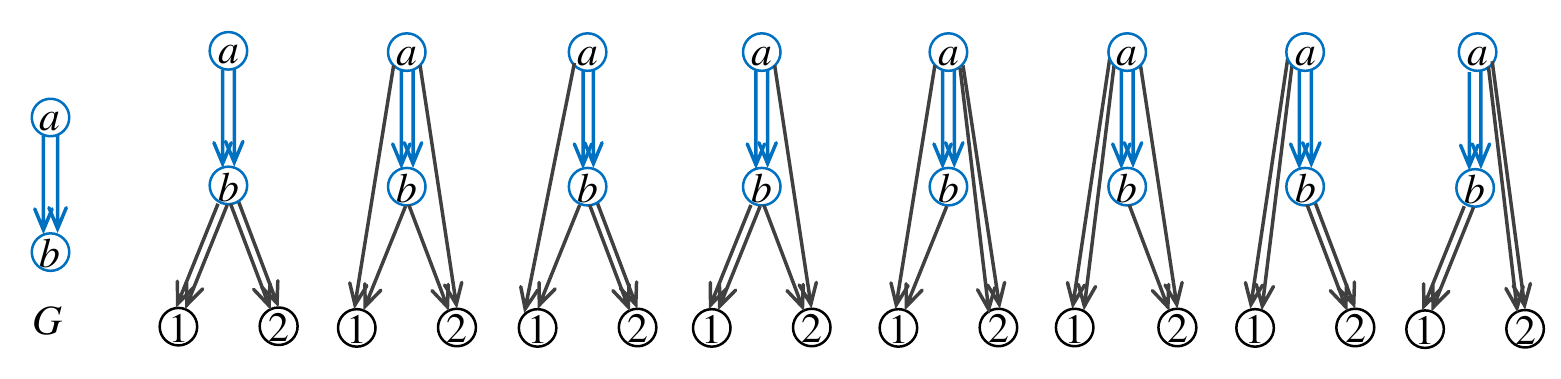}
           \caption{{\bf Illustration of graph extension.} $G$ (blue) consists of two parallel edges from  node $a$ to node $b$. 
   It can be extended into  eight non-isomorphic  labeled DAGs of level 3 by adding two new leaves. }
            \label{Fig4_GExtend}
\end{figure}

\begin{figure}%[t!]
            \centering
            \includegraphics[scale = 1.0]{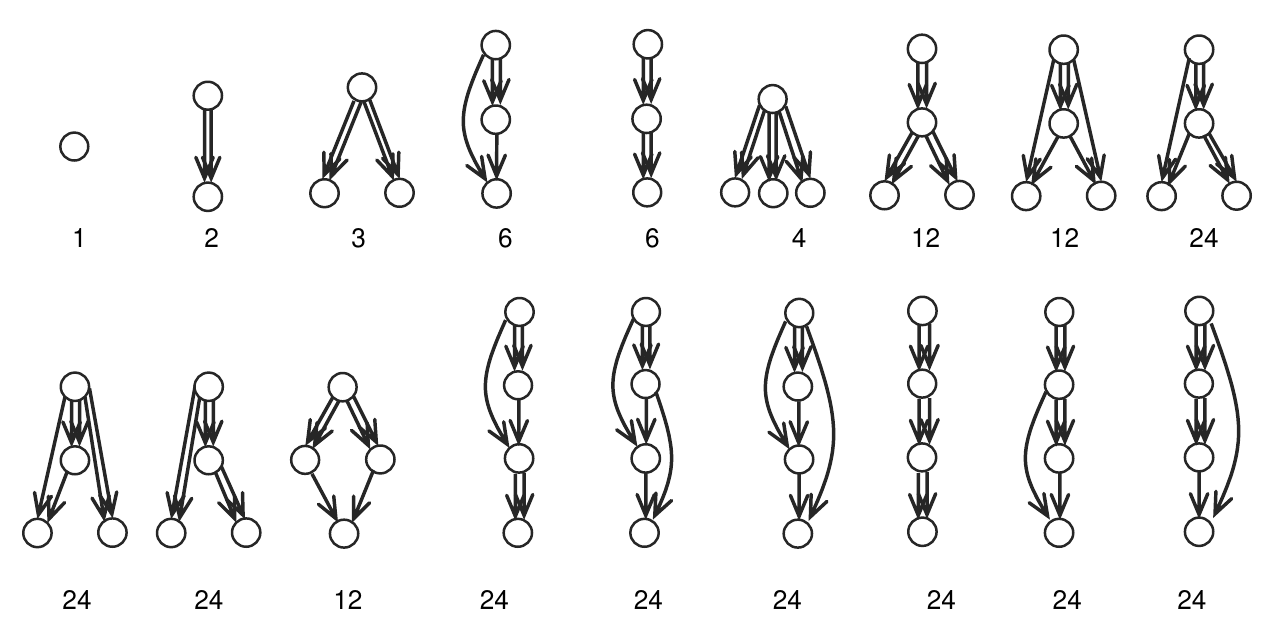}
           \caption{{\bf List of 18 unlabeled component graphs with 1 to 4 nodes.} The graphs are listed in increasing order according to level, in which the nodes are arranged row by row.   The number below each structure is the number of corresponding labeled component graphs. }
            \label{Fig1_Appendix}
\end{figure}

\begin{figure}%[t!]
            \centering
            \includegraphics[scale =0.70]{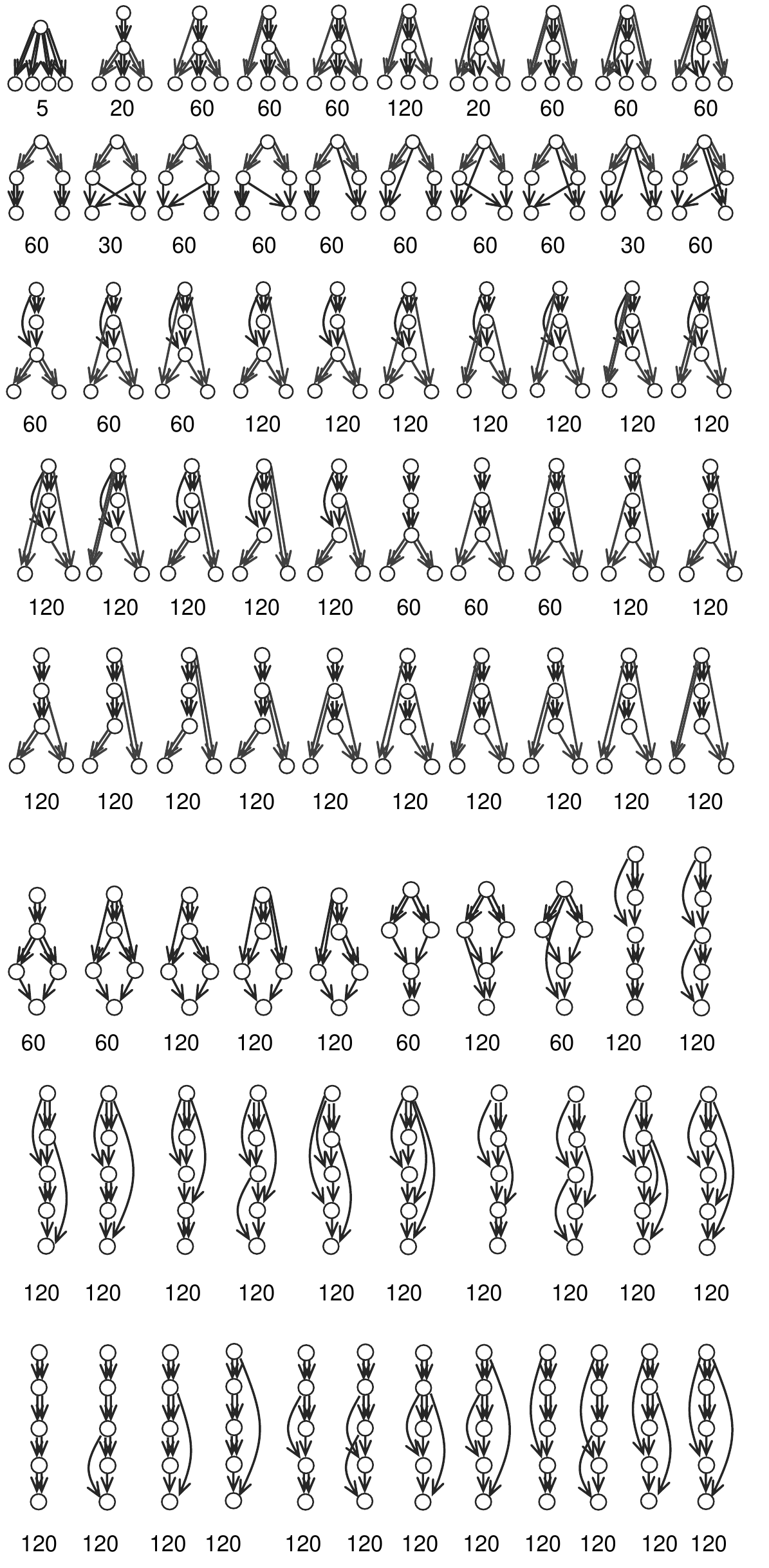}
           \caption{{\bf List of 82 unlabeled component graphs with 5 nodes.} The graphs are listed in increasing order according to level. The nodes of each graph are arranged row by row. The number below each structure is the number of corresponding labeled component graphs. There are 10 structures in all but the last row. }
            \label{Fig2_Appendix}
\end{figure}

\begin{theorem} 
 ${\cal D}_m$ denotes the set of labelled rooted DAGs in which the non-root nodes are of indegree 2 and double edges between two nodes are allowed.  Let $\alpha_m=|{\cal D}_m|$ and
$\alpha_m(s)=|\{ G\in {\cal D}_m \;|\;  |R_0(G)|=s\}|$.  The counts $\alpha_m$ and $\alpha_m(s)$ can  be computed via the following recurrence relations:
  \begin{eqnarray}
&& \alpha_1( 1)=1, \nonumber \\
&&  \alpha_m=\sum_{1\leq s\leq m-1} \alpha_m(s), m>1 \label{eqn_DAP_level}\\
&& \alpha_m(s) =  \sum_{1\leq t\leq \min (s/2, m-s-1)} {m \choose s} \beta(m, s, t)\alpha_{m-s}(t),   \label{eqn_DAP_level2}
\end{eqnarray}
where   
\begin{eqnarray}
 \beta(m, s, t)=\sum_{0\leq \ell \leq t} (-1)^{\ell}{t \choose \ell} {m-s-\ell+1\choose 2}^s
 \label{eqn_DAP_level3}
\end{eqnarray}
and we assume ${1\choose 2}=0$. 
\end{theorem}
%{\bf Proof}
\begin{proof}
Both  $\alpha_1(1)=1$ and Eqn.~(\ref{eqn_DAP_level}) are straightforward. 

To prove Eqn.~(\ref{eqn_DAP_level2}), we consider how the edges are added between the new nodes $v_i$ and the nodes in $G$.  For each new node $v_i$,  the two edges entering  $v_i$ will be added from a common node  or from two different nodes $u_j$ and $u_k$ ($k\neq j$). Therefore, 
 all the $s$ new nodes  can be connected to $G$ in $\left({y \choose 2}+y\right)^s= {y+1\choose 2}^s$ ways if the edges have  $y$ possible tails  for $y\leq m-s$.    By the Principle of Inclusion--Exclusion, the constraint that $u_i$ must be connected to some $v_j$ for each $i\leq t$ implies that the nodes in $G$ are connected to $u_i$'s in $ \beta(m, s, t)$ ways.

Since each node $u_i$  is of indegree 2, the constraint that each $v_i\in R_0(G)$ has to be connected to at least one new node implies that 
$t\leq s/2$. Moreover, $t$ is the number of leaves in $G$. Since $G$ has  $m-s$ nodes, $t\leq m-s-1$, where $m-s>1$.
For each of the ${m\choose s}$ labelings of the new notes $v_i$, there are $\alpha(m-s, t)$ DAGs to be extended. 
Taken together, these two facts imply the formula in Eqn.~(\ref{eqn_DAP_level3}). %$\QED$
\qed\end{proof}
%\begin{eqnarray}
% \alpha_m(s) =  \sum_{1\leq t\leq \min (s/2, m-s-1)}  {m \choose s} \alpha_{m-s}(t)\beta(m, s, t).
%\end{eqnarray}

\begin{remark}
(a)  The count $\alpha_m$ is actually the number of all labeled rooted RPNs in which every tree node is of outdegree 2 or more  and each reticulation is of exact indegree 2. These networks are called 
{\it bicombining RPNs} in literature \cite[page 140]{Huson_book}. 

(b)  Figures~\ref{Fig1_Appendix} and \ref{Fig2_Appendix}  list actually  all the unlabeled component graphs of all RPNs with zero  to four reticulations. However, for a RPN that is not a tree-child network, each tree-component of it may or may not contain any network leaves  and some components may be empty if there exist adjacent reticulations. We will use this fact for counting galled networks and arbitrary RPNs with only two reticulations in Section~\ref{section5.3}. 
\end{remark}

\subsection{Arbitrary TCNs}

Now we are able to  enumerate all the TCNs on $[n]$ by further extending 
all the component graphs in $\cup_{1\leq m\leq n}{\cal D}_{m}$.  

An $m$-{\it partition} of a set is a partition of the set into exactly $m$ non-empty parts.  
Let $1\leq m< n$ and $\pi$ be a partition of $[n]$ that divides $[n]$ into $m+1$ non-empty parts, say, 
$\{B_i\}^{m}_{0}$. We consider  all the $\alpha_{m+1}$ graphs $G_j$ in
${\cal D}_{m+1}$.  We further assume that all the graphs in ${\cal D}_{m+1}$ have nodes labeled by integers from 0 to $m$.  
We  extend all $G_j$'s into TCNs by reversing the network compression process presented in Table~\ref{algorithm}.\\

\begin{table}[!t]
\label{algorithm}
\caption{An algorithm for enumeration of tree-child networks.}
%\begin{center}
{\small
\begin{tabular}{|l|}
\hline 
    \hspace*{8em}{\sc TCN Enumeration Algorithm}\\
\\
   {\bf Input}: A $(m+1)$-partition $\pi$ with partition blocks $\{B_i\}^{m}_{0}$ and ${\cal D}_{m+1}$;\\
   {\bf Output}:  All the TCNs with $m$ reticulations on $[n]$ extending from the graphs in ${\cal D}_{m+1}$;\\
\\ 
   ${\cal TC}({\cal}{D}_{m+1}, \pi)=\emptyset$;   /*{\tt the set of TCNs extending from $G$s in
 ${\cal D}_{m+1}$ using $\pi$} */ \\
   {\bf for} each ordered list  of $(m+1)$ phylogenetic trees $T_i$ such that $T_i$ is on $B_i$, {\bf do} \\
\hspace*{1em}{\bf for} each $G\in {\cal D}_{m+1}$ with $m+1$ nodes $v_i$ with label  $i$ and outdegree $d_i$, {\bf do} \{\\
   \hspace*{2em} S0. Replace the root $v_0$ with $T_0$;\\
  \hspace*{2em} S1. If $i>0$, change $v_i$ to a reticulation $r_i$  and attach the tree $T_i$ below $r_i$ \\
  \hspace*{4.5em}by identifying the root of $T_i$ with the child of $r_i$; \\
   \hspace*{2em} S2. For each $i$, exhaustively select a {\it ordered} list of edges $\{e_j\}^{d_i}_1$ of $T_i$, where\\
  \hspace*{2em} ~~~~~~$e_i$'s can be identical,  and {\bf do} \{\\
  \hspace*{4.5em}  S2.1. Insert the tail of $j$-th edge leaving $v_i$ into the corresponding edges $e_j$,\\
 \hspace*{4.5em} ~~~~~~ where if $e_j=e_k$, the relative positions of the tails of the two edges\\
  \hspace*{4.5em} ~~~~~~ will be considered;\\
\hspace*{4.5em} S2.2.  Add the resulting TCN into ${\cal TC}({\cal}{D}_{m+1}, \pi)$;\\
 \hspace*{4.5em} ~~~~~~ /* multiple graphs will be added at this step */\\
  \hspace*{4em} \} /* end do in S2 */\\
  %\hspace*{2em} S3.  Add the resulting TCN into ${\cal TC}({\cal}{D}_{m+1}, \pi)$;\\
\hspace*{1.5em}\} /* end for */\\
\} /* end for */\\
 \hline
\end{tabular}
}
\end{table}
%\end{center}
\vspace{1em}

\begin{theorem}
Let $\gamma_n(m)$ denote the number of TCNs with $m$ reticulations on $[n]$ and let
$\Pi_{n, m+1}$ be the set of the $(m+1)$-partitions of $[n]$. The count $\gamma_n(m)$ can be computed via the following formula:
\begin{eqnarray}
\label{tcn_formula}
   \gamma_n(m) =\frac{1}{2^{n+m-1}}\sum_{\{B_i\}^{m}_{i=0}\in \Pi_{n, m+1}}\sum_{G\in {\cal D}_{m+1}} \prod_{i=0}^{m} \frac{2^{c_i}(2|B_i|+d_i-2)!}{(|B_i|-1)!},
\end{eqnarray}
where $d_i$ and $c_i$ are the number of the outgoing edges and the children of  each node $v_i$ of $G$, respectively.
\end{theorem}
%{\bf Proof.}
\begin{proof}
Let $G\in {\cal D}_{m+1}$ and $\pi=\{B_i\}^{m}_{0}$ be a $(m+1)$-partition of 
$[n]$. 
We also let $b_i=|B_i|$ for $0\leq i\leq m$.
There are $\frac{(2b_i-2)!}{2^{(b_i-1)}(b_i-1)!}$ phylogenetic trees $T_i$ on $B_i$ for expanding $v_i$. 
Since each phylogenetic tree on $B_i$ has $2b_i-1$ edges and the number of tree edges increases by 1 after an edge is inserted in the tree,
the $d_i$ edges leaving $v_i$ can be inserted in the tree in $(2b_i-1)(2b_i)\cdots (2b_i+d_i-2)$ ways.
However, if two edges $e'$ and $e''$ leaving $v_i$ lead to the same node,  then inserting $e'$ in an edge 
$x$ and inserting $e''$ in another edge $y$ produces the same tree as inserting $e'$ in $y$ and inserting $e''$ in $x$.
Let $c_i=\{x \;|\; x\in {\cal V}(G): (v_i, x)\in {\cal E}(G)\}$. Then, there are parallel edges between $v_i$
and  $d_i-c_i$ neighbors.

Since $\sum_{0\leq i\leq m}b_i=n$ and $\sum_{0\leq i\leq m}d_i=2m$, we can count all the possible extensions of $G$ with $\pi$ as:
\begin{eqnarray*}
    \gamma_n(m)&=&\prod_{j=0}^{m} \frac{(2b_i-2)!}{2^{(b_i-1)}(b_i-1)!}\cdot \frac{(2b_i-1)(2b_i)
\cdots (2b_i+d_i-2)}{2^{d_i-c_i}}\\
   &=& \frac{1}{2^{n+m-1}} \prod_{j=0}^{m} \frac{2^{c_i}(2b_i+d_i-2)!}{(b_i-1)!},
\end{eqnarray*}
implying Eqn.~(\ref{tcn_formula}).  %$\QED$
\qed\end{proof}

The number of TCNs with $k$ reticulations on $[n]$ for $1\leq k < n$ and $3\leq n\leq 8$ are computed and listed in Table~\ref{Table2_TCN}. This results can be checked using the python and SAGE \cite{SAGE} scripts that can be dowloaded from
\url{https://github.com/bielcardona/BTC_fixed_h}.
The computations in this table suggested the following result.

%{\scriptsize
\begin{table}
\centering 
\caption{Counts of TCNs with $k$ reticulations on $[n]$,  where $1\leq k <n$ and $3\leq n\leq 8$.
The last row contains the total numbers of TCNs that are not  phylogenetic trees, where the 
counts for $n\leq 7$ were  first obtained by Cardona et al. in \cite{Cardona_19}.
\label{Table2_TCN}
}
{\small
\begin{tabular}{c|rrrrrrr}
 \hline
  $k$\textbackslash $n$ & 2& 3 & 4& 5 & 6 & 7 & 8 \\ 
\hline
 1   &2 &  21 & 228 & 2,805    &39,330 & 623,385                  & 11,055,240 \\
 2   & & 42 & 1,272 & 30,300 &696,600  &16,418,430              &  405,755,280\\
3    & &     &  2,544 & 154,500   &6,494,400  & 241,204,950     & 8,609,378,400\\
4    & &     &          &309,000    &31,534,200   & 2,068,516,800 &  113,376,463,200\\
5    & &     &         &               &63,068,400 &9,737,380,800     &920,900,131,200 \\
6   &  &     &         &              &                &19,474,761,600     & 4,242,782,275,200 \\
7   &  &     &        &               &                &                          & 8,485,564,550,400\\
total & 2& 63 & 4,044 & 496,605 & 101,832,930 & 31,538,905,965 & 13,769,649,608,920 \\
\hline
\end{tabular}
}
%\\
%$^a$  The number of non-tree tree-child networks.
\end{table}

\begin{proposition}
  \label{thm:maximal}
  For each integer $n\ge3$, the number of TCNs with maximum number of reticulations (i.e. $k=n-1$) is twice the number of those with one less reticulation. In symbols, $|{\cal TC}_{n, n-1}|=2|{\cal TC}_{n, n-2}|$.
\end{proposition}

\begin{proof}
%  XXXXXXXX
  Let $N$ be a TCN on $[n]$ with $n-1$ reticulations and $\rho$ be its root. Note that it implies that each of its tree components has a single leaf, that is, it is reduced to a path ending in a leaf. From this we derive that  $\rho$ has  two grandchildren (in $N$) and that  one of the grandchildren  is a tree node, say $u$,  and the other is a reticulation, say $v$. If we remove from $N$ the arc entering $v$ from the child of $\rho$, both the child of $\rho$ and $v$ are of indegree $1$ and outdegree $1$. After elimination of $v$ and the child of $\rho$ by edge contraction we get a TCN $N'$ with $n-2$ reticulations. In this way we get a mapping $\psi: {\cal TC}_{n,n-1}\to {\cal TC}_{n,n-2}$.

  Let $N'$ be a TCN network on $[n]$ with $n-2$ reticulations and $\rho'$ be its root. Then, there is  a single tree component of $N'$ that contains  two leaves, and all the other ones are simply paths. Hence, there exists a single node $u$ in the forest $N'-{\cal R}(N')$ with two children $w_1$ and $w_2$, which are of indegree $1$ (and outdegree either $0$ or $1$) for each $i$. Let $w=w_i$ for $i=0, 1$.  Make a subdivision of the arc $(u,w)$ by
  introducing a node $v$ and a subdivision of the arc $(\rho', c(\rho'))$ by introducing a new node $p$ and then add the arc $(p,v)$. The node $p$ is the child of $\rho'$ in the obtained network, which is a TCN, and $v$ is a new reticulation. Hence we get a network $N$ with $n-1$ reticulations, and it is straightforward to check that $\psi(N)=N'$. Also, since there are exactly two choices for $w$, and they produce non-isomorphic networks, it follows that the mapping $\psi$ is exhaustive and each network with $n-2$ reticulations has two preimages. 
\qed\end{proof}

\section{Counting TCNs with Few Reticulations}
\label{section5}

\subsection{Relationships between network classes}
\label{section5.1}

\begin{proposition}
\label{1ret_case}
  Let  $n>2$. Then, $${\cal NN}_{n, 1}\subset {\cal TC}_{n, 1}={\cal GT}_{n, 1}={\cal GN}_{n, 1}={\cal RV}_{n, 1}={\cal RPN}_{n, 1}.$$
\end{proposition}
\begin{proof} 
These set proper containment and equations can be derived from the definitions of these classes. 
\qed\end{proof}

\begin{proposition} 
\label{2ret_case}
Let $n>2$.  
\begin{enumerate}
\item ${\cal NN}_{n, 2}\subset  {\cal TC}_{n, 2}$. 

\item ${\cal GT}_{n, 2}\subset  {\cal GN}_{n, 2} \subset  {\cal RV}_{n, 2} \subset {\cal TB}_{n, 2}$.

\item  ${\cal TC}_{n, 2} \cup {\cal GN}_{n, 2} \subset {\cal RV}_{n, 2}$.

\item  ${\cal TB}_{n, 2}={\cal RPN}_{n, 2}$.
\end{enumerate}
\end{proposition}
\begin{proof} The first three relationships are straightforward. For example, the second one is deduced from 
the facts that (a) the galled network in Figure~\ref{Fig_1}D is not a gall tree, (b) the network in Figure~\ref{Fig_1}E is a reticulation-visible network but not a gall network and (c) the network in Figure~\ref{Fig_1}F is a tree-based network but not a reticulation-visible network. 

  As for the last one, let $N\in {\cal RPN}_{n, 2}$.  Since $N$ contains only two reticulations, only one reticulation can have the other as its parent, implying that
 $N$ is tree-based (see \cite[Corollary 10.18, page 260]{Steel_book} or  \cite[Theorem 1]{Zhang_16}).
\qed\end{proof}

Recall that 1-${\cal C}$ denotes the subset of one-component networks of $\cal C$ for a network class ${\cal C}$. 

\begin{proposition} 
\label{prop5.3}
Let $n\geq 2$ and $k\geq 1$ such that $n\geq k$.  
\begin{enumerate}
\item $\mbox{1-}{\cal NN}_{n, k} \subset  \mbox{1-}{\cal TC}_{n, k}$. 
\item $\mbox{1-}{\cal GT}_{n, k} \subset  \mbox{1-}{\cal TC}_{n, k}$.
\item (\cite{Zhang_Rathin_G18}) $\mbox{1-}{\cal TC}_{n, k} \subset  \mbox{1-}{\cal GN}_{n, k} =  \mbox{1-}{\cal RV}_{n, k} = \mbox{1-}{\cal TB}_{n, k} = \mbox{1-}{\cal RPN}_{n, k}$.
\end{enumerate}
\end{proposition}
\begin{proof}
  The relationships in the different items can be derived from the definition of these classes.
%For example, (iii) can be deduced from 
%the facts that  (a) the galled network in Figure~\ref{Fig_1}D is not a TCN and 
%(b) every reticulation $r$ is inner for one-component networks, 
%as the  two parents of $r$ are in the same tree-component that contains the network root.
\qed\end{proof}

\subsection{Counting RPNs with one reticulation}
By Proposition~\ref{1ret_case},  the hierarchy of network classes beyond galled trees collapses  from five  into one. 
Recently, the author of this paper obtained simple formulas for the size of ${\cal NN}_{n, 1}$ and the size of ${\cal RPN}_{n, 1}$ (see \cite{Zhang_19}).  The formula for ${\cal RPN}_{n, 1}$ will be used in the next subsection.  For completeness,  the formula for the count of normal networks is also given below.

\begin{proposition} [\cite{Zhang_19}]
\label{oneret_case}
 Let $n\geq 3$. Then, 
 \begin{eqnarray}
  && |{\cal NN}_{n, 1}|=\frac{(n+2)(2n)!}{2^{n}n!}-3\cdot 2^{n-1}n!, \label{1ret_formulaN}\\
  && |{\cal RPN}_{n, 1}|= \frac{n(2n)!}{2^n n!}-2^{n-1} n!. \label{1ret_formulaArb}
\end{eqnarray}
\end{proposition}

%We remark that explicit formulas for the number of unrooted and rooted networks with one reticulation were presented in
%\cite{Bouvel_18}.  But the formulas for  $|{\cal RPN}_{n, 1}|$ presented here is much simpler than %theirs \cite[Theorem 6]{Bouvel_18}.

\subsection{Counting RPNs with two reticulations}
\label{section5.3}

Proposition~\ref{2ret_case} implies that all the six network classes defined in Section~\ref{sec2.2} are distinct. This raises different counting problems for 
RPNs with two reticulations. In the resit of  this section, we will answer three of them. 

\begin{lemma} 
\label{lemma55}
For $n\geq 2$,
\begin{eqnarray}
  && \sum^{n}_{k=1}{2k\choose k} \frac{k}{2^{2k}}=\frac{(2n+1)!}{3\cdot 2^{2n}n!(n-1)!}, \label{sum1}\\
  && \sum^{n-1}_{k=1}{2k\choose k}{2n-2k\choose n-k}k(n-k)=n(n-1)2^{2n-3}, \label{identity1}
 % && \sum^{n}_{k=1}{2k\choose k}{2n-2k\choose n-k}k=n2^{2n-1}. \label{identity2}
\end{eqnarray}
\end{lemma}
\begin{proof}
 Eqn.~(\ref{sum1}) is trivial for $n=1$. Assuming that  it is true for $n-1$, we then have:
\begin{eqnarray*}
  \sum^{n}_{k=1}{2k\choose k} \frac{k}{2^{2k}}= \frac{(2(n-1)+1)!}{3\cdot 2^{2(n-1)}(n-1)!(n-2)!} +{2n\choose n} \frac{n}{2^{2n}} = \frac{(2n+1)!}{3\cdot 2^{2n}n!(n-1)!}.
\end{eqnarray*}
This proves Eqn.~(\ref{sum1}).

Let $f(x)=\sum_{n\geq 0} {2n\choose n}x^n$. Then, $f(x)=(1-4x)^{-1/2}$ (see \cite[page 52]{Stanley_book}). Multiplying  $f'(x)$ by $x$, we have:
$$xf'(x)=\sum_{n\geq 0} {2n\choose n} n x^{n},$$
and 
$$(xf'(x))^2 =\sum_{n\geq 0} \left(\sum^n_{k=0}{2k\choose k}{2n-2k\choose n-k}k(n-k)\right) x^n.$$
On the other hand, 
$$(xf'(x))^2 =x^2 \left( 2(1-4x)^{-3/2}\right)^2 =  \sum_{n\geq 0}{2+n\choose 2}4^{n+1}x^{n+2}.$$
Identifying the coefficient of $x^n$ in these two forms, we obtain Eqn.~(\ref{identity1}).
\qed\end{proof}

\begin{theorem} 
\label{theorem_tworet}
Let $n>2$. The number of tree-child networks with two reticulations on $n$ taxa is
{\small
\begin{eqnarray*}
 |{\cal TC}_{n, 2}| &=& \frac{n!}{2^{n}} \sum^{n-2}_{j=1}{2j\choose j}{2n-2j\choose n-j} \frac{j(2j+1) (2n-j-1)}{2n-2j-1} \\
  & &  + n(n-1)n!2^{n-3} -\frac{(2n-1)!n}{3\cdot 2^{n-1}(n-2)!}.
\end{eqnarray*}
}
\end{theorem}
\begin{proof}
There are three different component graphs (the third to fifth graphs in Figure~\ref{Fig1_Appendix}) for TCNs with two reticulations, called $G_3, G_4$ and $G_5$. 
Let $A_i$ be the number of TCNs having $G_i$ as their component graph for $i=3, 4, 5$. 

Consider a TCN such that its component graph is $G_3$.  The structure of $G_3$ suggests that the top tree-component of the TCN  is a one-component TCN with a fixed reticulation, whereas  the bottom two tree-components are both a phylogenetic tree with the leaves contained in these components. If the top tree-component contains $j$ network leaves, the bottom tree-components correspond to a forest of two phylogenetic trees on $n-j$ taxa and hence there are  $t_{n-j}$ possibilities in total,  where $t_{n-j}$ is  the number of all the phylogenetic trees with $n-j$ taxa. Applying the same argument as in the proof of Theorem 4.7 to the top tree-component,  we obtain:
\begin{eqnarray*}
 A_3
&=& \sum^{n-2}_{j=1} {n\choose j} \frac{(2j+2)!}{2^{j+1}(j-1)!} t_{n-j}
= \frac{n!}{2^n} \sum^{n-2}_{j=1}  \frac{(2j+2)!}{j!(j-1)!}\frac{(2n-2j-2)!}{(n-j-1)!(n-j)!}
\end{eqnarray*}

Consider a TCN such that $G_4$ is its component graph. We have then that the bottom tree-component of the TCN is a phylogenetic tree 
with at least one leaf, the middle tree-component is a phylogenetic tree with $k+1$ leave if it contains $k$ leaves. Thus, if  the top tree-component contains $j$ leaves,  the bottom two tree-components is essentially a phylogenetic tree with $n-j$ leaves with an edge from the top tree-component being inserted into  one of $2n-2j-2$ tree edges that are not adjacent to the top reticulation node.  Thus, 
\begin{eqnarray*}
     A_4
%&=&\sum^{n-2}_{j=1}\sum^{n-j-1}_{k=1} {n\choose j,\;k,\;n-j-k} \frac{(2j+1)!}{2^{j}(j-1)!}a_{k+1}a_{n-j-k}\\
%&=&\sum^{n-2}_{j=1}\sum^{n-j-1}_{k=1} {n\choose j,\;k,\;n-j-k} \frac{(2j+1)!}{2^{j}(j-1)!}\frac{(2k-1)!}{2^{n-j-2}
%(k-1)!}\frac{(2(n-j-k)-2)!}{(n-k-j-1)!}\\
%&=&\frac{1}{2^{n-2}}\sum^{n-2}_{j=1}\sum^{n-j-1}_{k=1}{n\choose j,\;k,\;n-j-k} \frac{(2j+1)!}{(j-1)!}%\frac{(2k-1)!}{(k-1)!}\frac{(2(n-j-k)-2)!}{(n-k-j-1)!} \\
%&=&\frac{n!}{2^{n-2}} \sum^{n-2}_{j=1}\sum^{n-j-1}_{k=1} \frac{(2j+1)!}{j!(j-1)!}\frac{(2k-1)!}{k!(k-1)!}C_{n-k-j-1}\\
 &=&\sum^{n-2}_{j=1} {n \choose j} \frac{(2j+1)!}{2^{j}(j-1)!}\times \frac{(2n-2j-2)!}{2^{n-j-1}(n-j-1)!}\times (2n-2j-2)\\
&=& \frac{n!}{2^{n-1}}\sum^{n-2}_{j=1} \frac{(2j+1)!}{j!(j-1)!}\frac{2(n-j-1)(2n-2j-2)!}{(n-j-1)!(n-j)!},
\end{eqnarray*}
where in the first row, the first term is the number of possibilities for the top tree-component, the second term is   the number of possibilities for the tree structure contained in the bottom two tree-components and  the third term is the number of 
possibilities of forming the lower reticulation by inserting the fixed leave of the top tree-component into an edge in the bottom two tree-components. 

Consider a TCN such that its component graph is $G_5$. The bottom two tree components form  a TCN with one reticulation, whereas the top component is a one-component TCN with a fixed reticulation. Thus, by Proposition~\ref{oneret_case} and Lemma~\ref{lemma55},   we have:
\begin{eqnarray*}
     A_5
&=&  \sum^{n-2}_{j=1} {n \choose j}  \frac{(2j)!}{2^{j}(j-1)!} \left( \frac{(n-j)(2n-2j)!}{2^{n-j}(n-j)!} -2^{n-j-1}(n-j)!\right)\\
&=& \sum^{n-2}_{j=1} {n \choose j}  \frac{(2j)!}{2^{j}(j-1)!}\frac{(2n-2j)!}{2^{n-j}(n-j-1)!} -
 \sum^{n-2}_{j=1} {n \choose j}  \frac{(2j)!}{2^{j}(j-1)!} 2^{n-j-1}(n-j)!\\
&=&\frac{n!}{2^n} \sum^{n-2}_{j=1} {2j\choose j}{2n-2j\choose n-j}j(n-j) -
 n! 2^{n-1} \sum^{n-2}_{j=1} \frac{(2j)!}{4^j j!(j-1)!}\\
&=&\frac{n!}{2^n} \sum^{n-2}_{j=1} {2j\choose j} {2n-2j \choose n-j} j(n-j) -
 n! 2^{n-1} \sum^{n-2}_{j=1}{2j\choose j} \frac{j}{4^j}\\
%&=& \sum^{n-2}_{j=1} {n \choose j}  \frac{(2j)!}{2^{j}(j-1)!}\frac{(2n-2j)!}{2^{n-j}(n-j-1)!}
% -\frac{(2n-2)!n}{3\cdot 2^{n-2}(n-3)!}\\
%&=& \frac{n!}{2^n}\left( n(n-1)2^{2n-3} -{2n-2\choose n-1}{2\choose 1}(n-1)\right)
% -\frac{(2n-2)!n}{3\cdot 2^{n-2}(n-3)!}\\
%&=& n(n-1)n!2^{n-3}-\frac{(2n-2)!n}{2^{n-1}(n-2)!} -\frac{(2n-3)!n(n-1)}{3\cdot 2^{n-3}(n-3)!}\\
&=& n(n-1)n!2^{n-3}-\frac{(2n-1)!n}{3\cdot 2^{n-1}(n-2)!}.\\
\end{eqnarray*}

Summing the above three equations, we obtain the total number of TCNs with two reticulations on $n$ taxa:
\begin{eqnarray*}
   |{\cal TC}_{n, 2}|   &=&A_3+A_4+A_5\\
&=&  \frac{n!}{2^n} \sum^{n-2}_{j=1}  \frac{(2j+2)!}{j!(j-1)!}\frac{(n-j)(2n-2j-2)!}{(n-j)!(n-j)!}\\
 & & +\frac{n!}{2^{n-1}}\sum^{n-2}_{j=1} \frac{(2j+1)!}{j!(j-1)!}\frac{2(n-j-1)(n-j)(2n-2j-2)!}{(n-j)!(n-j)!} +A_5\\
%&=&  \frac{n!}{2^{n-2}} \sum^{n-2}_{j=1}  \frac{(2n^2+n-2j-1)(2j)!}{j!(j-1)!}\frac{(2n-2j-2)!}{(n-j)!(n-j)!}-\frac{(2n-3)!n(n-1)}{3\cdot 2^{n-3}(n-3)!}\\
&=&  \frac{n!}{2^{n}} \sum^{n-2}_{j=1}{2j\choose j}{2n-2j\choose n-j} \frac{j(2j+1) (2n-j-1)}{2n-2j-1} +A_5. %\\
%&=& n(n-1)n!2^{n-3} + \frac{n!}{2^{n}} \sum^{n-2}_{j=1}{2j\choose j}{2n-2j\choose n-j} \frac{j(2j+1) (2n-j-1)}{2n-2j-1} -\frac{(2n-1)!n}{3\cdot 2^{n-1}(n-2)!}.\\
\end{eqnarray*}
The proposition is proved.
\qed\end{proof}

Let $g_{n, k}$ denote the number of one-component networks with $k$ reticulations on $[n]$ in which the children of the $k$ reticulations are Leaf $1, 2, \cdots k$, $k\geq 0$.  Notice that  $g_{n, 0}$ is simply the number of phylogenetic trees on $[n]$.
By  Eqn. (4), (5) and (10) in \cite{Zhang_Rathin_G18}, we have the following formula: 
\begin{eqnarray}
    %g_{n, 1}&=& \frac{(2n-2)!}{2^{n-1}(n-2)!}, \label{1ret_formula}\\
    g_{n, 2}&= &ng_{n, 1} + \frac{1}{2} \left(g_{n-1, 0} + g_{n, 0}\right)=\frac{(n-1)^2(2n-1)(2n-4)!}{2^{n-2}(n-2)!}, \label{2ret_formula}
\end{eqnarray}
where we assume $0!=1$.

\begin{proposition} 
\label{prop5.6}
Let $n\geq 2$ and let ${\cal GN}_{n, 2}$ be the set of galled networks with two reticulations on $n$ taxa. Then,
  \begin{eqnarray*}
    |{\cal GN}_{n, 2}|
&=&  \frac{n!}{2^{n-1}}\sum^{n-2}_{j=0} {2j \choose j}{2n-2j \choose n-j}\frac{ (j+1)^2 (2j+3)}{(n-j)(2n-2j-1)} \\
& &+ n(n-1)n!2^{n-3} -\frac{(2n-1)!n}{3\cdot 2^{n-1}(n-2)!} \\
\end{eqnarray*}
\end{proposition}
\begin{proof} Since the component graph of a galled network with two reticulations is a tree with three nodes. Thus,  the component tree is either $G_3$ or $G_5$ in Figure~\ref{Fig1_Appendix}. However, as we remarks in Section 4.2,  unlike a TCN, 
a component may not contain any leaf if it is not at the bottom for a galled network.  

For a galled network with $G_3$ as its component graph,  the top tree-component is a 1-component network 
with two distinguished reticulations that may or may not contain network leaves. Additionally,  each of its two bottom components is a phylogenetic tree with at least one leaf, equivalent to a phylogenetic tree with them as left and right subtrees. Therefore, by Eqn~(\ref{2ret_formula}), the number of galled networks for this case is:
\begin{eqnarray*}
   B_3&=& \sum^{n-2}_{j=0} {n\choose j} g_{j+2, 2}\cdot g_{n-j, 0}\\
        &=& \sum^{n-2}_{j=0} {n\choose j} \frac{(j+1)^2(2j+3)(2j)!}{2^{j}(j)!}\frac{(2n-2j-2)!}{2^{n-j-1}(n-j-1)!}\\ 
   &=& \frac{n!}{2^{n-1}}\sum^{n-2}_{j=0} {2j \choose j}{2(n-j-1)\choose (n-j-1)}\frac{ (j+1)^2 (2j+3)}{(n-j)}\\
   &=& \frac{n!}{2^{n}}\sum^{n-2}_{j=0} {2j \choose j}{2n-2j\choose n-j}\frac{ (j+1)^2 (2j+3)}{(2n-2j-1)}.
\end{eqnarray*}

Consider a galled network. If its component graph is $G_5$,  the top tree-component is  then a 1-galled network with one distinguished reticulation  that contains at least  one network leaf; and the bottom tree-components form a galled network with one reticulation.  Since every TCN is a galled network if it contains only one reticulation (Proposition~\ref{1ret_case}), 
the number of galled networks with two reticulations that have $G_5$ as the component graph is equal to $A_5$ that was calculated in  the proof of Theorem~\ref{theorem_tworet}.

Summing $B_3$ and $A_5$, we obtain the formula. 
\qed\end{proof}

\begin{proposition} 
Let $n\geq 3$ and ${\cal GT}_{n, 2}$ be the set of galled trees with two reticulations on $n$ taxa. Then, 
{\small
\begin{eqnarray*}
\label{2ret_galls}
  |{\cal GT}_{n, 2}|=\frac{n!}{ 3\cdot 2^n}\sum^{n-2}_{j=1} {2j\choose j}{2n-2j\choose n-j} 
 \frac{j(j+1)(2j+1)}{(2n-2j-1)} + n(n-1)n!2^{n-3}-\frac{(2n-1)!n}{3\cdot 2^{n-1}(n-2)!}.
\end{eqnarray*}
}
\end{proposition}
\begin{proof}
Since galled trees are galled networks,  the component graph of  each galled tree with two reticulations  are either  $G_3$ and $G_5$ in Figure~\ref{Fig1_Appendix}. Since a galled tree is also a TCN, 
a component of a galled tree contains at least one network leaf.  

For a galled tree with $G_3$ as its component graph,  its top tree-component is a 1-galled tree
with two distinguished reticulations; and  each of its bottom two tree-components is a phylogenetic tree with at least one leaf. Therefore, by Eqn.~(\ref{2ret_galledtree}), the number of galled trees for this case is: 
\begin{eqnarray*}
  C_3&=&\sum^{n-2}_{j=1} {n\choose j} \frac{(2(j+2)-2)!}{3\cdot 2^{j+2-1}(j+2-3)!} \frac{(2n-2j-2)!}{2^{n-j-1}(n-j-1)!}\\
 %  &=&\frac{n!}{3\cdot 2^n}\sum^{n-2}_{j=1} \frac{(2j+2)!}{j!(j-1)!} \frac{(2n-2j-2)!}{(n-j)!(n-j-1)!}\\
 &=&\frac{n!}{ 3\cdot 2^n}\sum^{n-2}_{j=1} {2j\choose j}{2n-2j\choose n-j} 
 \frac{j(j+1)(2j+1)}{(2n-2j-1)}.
\end{eqnarray*}

Consider a galled tree. If its component graph is $G_5$,  the top tree-component is  then a 1-galled tree with a distinguished reticulation that  contains at least  one network leaf; and the  bottom two tree-components form a galled tree with one reticulation.  Since every TCN is also a galled tree if it contains only one reticulation, 
the number of galled trees for this case is equal to $A_5$ that was calculated in 
the proof of Theorem~\ref{theorem_tworet}.
%\begin{eqnarray*}
%  C_5&=&\sum^{n-2}_{j=1} {n\choose j} \frac{(2(j+1)-2)!}{2^{j+1-1}(j+1-2)!} \left(\frac{(2n-2j)!}%{2^{n-j}(n-j-1)!} -2^{n-j-1}(n-j)! \right)\\
%&=& n!\sum^{n-2}_{j=1} \frac{(2j)!}{2^{j}j!(j-1)!(n-j)!} \left(\frac{(2n-2j)!}{2^{n-j}(n-j-1)!} 
%-2^{n-j-1}(n-j)! \right)\\
%&=& \frac{n!}{ 2^{n}}\sum^{n-2}_{j=1} \frac{(2j)!}{j!(j-1)!} \frac{(2n-2j)!}{(n-j)!(n-j-1)!} -{n!%2^{n-1}}\sum^{n-2}_{j=1} \frac{(2j)!}{2^{2j}j!(j-1)!} \\
%   &=&\frac{n!}{2^n}\sum^{n-2}_{j=1}{2j\choose j}{2n-2j\choose n-j}j(n-j) 
 %  - \frac{n(n-1)(2n-3)!}{3\cdot 2^{n-3}(n-3)!}.
%\end{eqnarray*}

Summing $C_3$ and $A_5$, we obtain the count result.
\qed\end{proof}

 Lastly, we point out that the following counting problems are open:
\begin{itemize}
  \item How to count normal networks with two reticulations on $n$ taxa?
  %\item How to count reticulation-visible networks with two reticulations on $n$ taxa?
  \item  How to count arbitrary RPNs with two reticulations on $n$ taxa?
\end{itemize}

\subsection{Counting one-component RPNs}

 Corollary~\ref{corollary3.6} presents a  formula for the count of one-component galled trees, while Theorem~\ref{thm1} presents a  formula for the count of one-component TCNs. 

Additionally, by Proposition~\ref{prop5.3},  the hierarchy of network classes beyond one-component galled networks collapses  from four  into one.   Since $\mbox{1-}{\cal RPN}_{n, k}=\mbox{1-}{\cal GN}_{n, k}$, the count 
$a_{n, k}$ of arbitrary one-component RPNs can be calculated via the following  recurrence formula \cite[Theorem 3]{Zhang_Rathin_G18}: 
\begin{eqnarray*}
 % &&c_{n}^{(k+1)}=(n+k-2)c_{n}^{(k)}+kc_{n}^{(k-1)} +\frac{k!}{2}\sum^{k}_{j=1}
 %  \frac{2^{-j}(2j)!}{ (k-j)! (j!)^2}\left(c^{(k-j)}_{n-j}-c^{(k-j)}_{n-j+1}\right), \nonumber\\
 % &&c_{n+1}^{(k+1)}=(n+k-1)c_{n+1}^{(k)}+kc_{n+1}^{(k-1)} +\frac{k!}{2}\sum^{k}_{j=1}
 %  \frac{2^{-j}(2j)!}{ (k-j)! (j!)^2}\left(c^{(k-j)}_{n+1-j}-c^{(k-j)}_{n-j+2}\right), \nonumber\\
%&&{n\choose k+1} c_{n+1}^{(k+1)}=\frac{(n-k)(n+k-1)}{(k+1)}\left({n\choose k} c_{n+1}^{(k)}
%+ {n\choose k-1} c_{n+1}^{(k-1)}\right) \\
%   && +\frac{n!(n-k)}{2(k+1)} \sum^{k}_{j=1}
%   \frac{2^{-j}(2j)!}{ (n-j)!(j!)^2}
%    \left({n-j\choose k-j}c^{(k-j)}_{n+1-j}-
%     \frac{n+1-k}{n+1-j}{n-j+1\choose k-j} c^{(k-j)}_{n-j+2} \right), \nonumber\\
a_{n, k+1}&=&\frac{(n-k)}{(k+1)}(n+k-1)\left( a_{n, k} + a_{n, k-1}\right) \\
   &&  +\frac{n!(n-k)}{2(k+1)} \sum^{k}_{j=1}
   {2j \choose j}\frac{
    (n+1-j) a_{n-j, k-j}-{(n+1-k)} a_{n+1-j, k-j}}{2^j (n+1-j)!}. \nonumber\\
 % && ~~~|\mbox{1-}{\cal TC}_{n, k+1}|={n \choose k+1}c_{n+1}^{(k+1)}.
\end{eqnarray*}

Lastly, the following problem is open:
\begin{itemize}
 \item How to count one-component normal networks with $k$ reticulations on $n$ taxa?
\end{itemize}

\section{Conclusion}
Only asymptotic counting results are known for the classes of TCNs and normal networks.  
We have presented few approaches and  formulas for the exact count of TCNs,  galled trees and galled networks.  In addition to the problems posed in Section 5,   the following questions are also open for future study: 
\begin{itemize}
  \item Is there  a simple closed or recurrence formula for the count of TCNs? 
  \item Is there a simple closed or recurrence  formula for the count  of galled trees? 
   \item Is there a simple closed or recurrence formula for the count of normal networks?
\end{itemize}
New approaches seem  to be needed to answer these three problems. 

In Section~\ref{Sect4_TC},  we have shown that a TCN is an expansion of 
a rooted, labeled DAG through the replacement of nodes with TCNs  in which 
the child of each reticulation is a leaf. We also prove that the replacement can be obtained from phylogenetic  trees by insertion of reticulations with a leaf child. Are these structural characterizations of TCNs useful in the study of other aspects of TCNs?

\section*{Authors' contributions}

This manuscript evolved from two different manuscripts, which were authored respectively by GC and LZ, and submitted independently to different journals. This coincidence became known to the Editor of one of the journals, who proposed making a joint paper with the results of both manuscripts. 
The present manuscript is a direct descendant of the paper authored by LZ, which except for Proposition~\ref{thm:maximal} contained all the results presented here. The manuscript authored by GC contained essentially the results found here in Sections~\ref{Sect4_TC} and (partially)~\ref{section5}, including also the aforementioned Proposition~\ref{thm:maximal}.

\section*{Acknowledgments}

GC thanks Francesc Rossell\'o for hinting the proof of Proposition~\ref{thm:maximal}.
LZ thanks Mike Steel for the discussion of this project and  Philippe  Gambette for sharing their manuscript on counting level-$k$ phylogenetic networks.

Research of GC was funded by
FEDER/Ministerio de Ciencia, Innovaci\'on y Universidades/Agencia Estatal de Investigaci\'on project PGC2018-096956-B-C43. Research of LZ was funded by Singapore's Ministry of Education Academic
Research Fund Tier-1 [grant R-146-000-238-114] and the National Research
Fund [grant NRF2016NRF-NSFC001-026].

%%main text



\bibliographystyle{elsarticle-num} 
\bibliography{Counting_My_references}


%\begin{thebibliography}{00}


%\bibitem{}

%\end{thebibliography}
\end{document}